\documentclass[journal]{IEEEtran}

\ifCLASSINFOpdf
\else
\fi

\usepackage{amssymb}
\usepackage{amsmath}
\usepackage{amsfonts}
\usepackage{bbm}
\usepackage{dsfont}
\usepackage{graphicx}
\usepackage[caption=false]{subfig}
\usepackage{cite}
\usepackage{multirow}
\usepackage{algpseudocode}
\usepackage{algorithm}
\usepackage{epstopdf}
\usepackage[utf8]{inputenc}
\usepackage[english]{babel}
\usepackage{color}
\usepackage{tikz}

\newtheorem{theorem}{Theorem}[section]

\newtheorem{lemma}[theorem]{Lemma}

\def\mb{\mathbf}

\def\tb {\textbf}

\def\be{\begin{equation}}
\def\ee{\end{equation}}
\def\ben{\begin{equation*}}
\def\een{\end{equation*}}
\def\bea{\begin{eqnarray}}
\def\eea{\end{eqnarray}}
\def\beaa{\begin{eqnarray*}}
\def\eeaa{\end{eqnarray*}}

\begin{document}


\title{Transmit MIMO Radar Beampattern Design Via Optimization on the Complex Circle Manifold }

\author{Khaled~Alhujaili,~\IEEEmembership{Student~Member,~IEEE,}
        Vishal~Monga,~\IEEEmembership{Senior~Member,~IEEE,}
        and~Muralidhar~Rangaswamy,~\IEEEmembership{Fellow,~IEEE}%
       \thanks{
       K. Alhujaili is with the Department of Electrical Engineering, Pennsylvania State University, USA and also with the Department of Electrical Engineering, Taibah University, Al-Madina 344, Saudi Arabia.
       V. Monga is with the Department of Electrical Engineering, Pennsylvania State University, USA.
       M. Rangaswamy is with the RF Exploitation Branch, US Air Force Research Lab, Dayton, Ohio, USA.
         Dr. Rangaswamy was supported by the Air Force Office of Scientific Research under project 2311IN.
         Research was supported by a grant from the Air Force Office of Scientific Research (AFOSR) Award No. FA9550-15-1-0438.}}

	\markboth{IEEE Transactions on Signal Processing, Accepted for Publication, April 2019.}%
{Shell \MakeLowercase{\textit{et al.}}: Bare Demo of IEEEtran.cls for IEEE Journals}
\maketitle

\vspace{-10mm}

\begin{abstract}
	The ability of Multiple-Input Multiple-Output (MIMO) radar systems to adapt waveforms across antennas allows flexibility in the transmit beampattern design. In cognitive radar, a popular cost function is to minimize the deviation against an idealized beampattern (which is arrived at with knowledge of the environment). The optimization of the transmit beampattern becomes particularly challenging in the presence of practical constraints on the transmit waveform. One of the hardest of such constraints is the non-convex constant modulus constraint, which has been the subject of much recent work. In a departure from most existing approaches, 
	we develop a solution that involves direct optimization over the non-convex complex circle manifold. That is, we derive a new projection, descent, and retraction (PDR) update strategy that allows for monotonic cost function improvement while maintaining feasibility over the complex circle manifold (constant modulus set). For quadratic cost functions (as is the case with beampattern deviation), we provide analytical guarantees of monotonic cost function improvement along with proof of convergence to a local minima. We evaluate the proposed PDR algorithm against other candidate MIMO beampattern design methods and show that PDR can outperform competing wideband beampattern design methods while being computationally less expensive. Finally,  orthogonality across antennas is incorporated in the PDR framework by adding a penalty term to the beampattern cost function. Enabled by orthogonal waveforms, robustness to target direction mismatch is also demonstrated.    
\end{abstract}
\begin{IEEEkeywords}
	MIMO radar, wideband beampattern, waveform design, constant modulus, complex circle manifold, manifolds, cognitive radar, PDR.
\end{IEEEkeywords}

\IEEEpeerreviewmaketitle

\section{Introduction}
\label{Sec:Introduction}

Multiple input multiple output (MIMO) radar, in general, transmits independent waveforms from its transmitting elements and observes the backscattered signals (from the target and from interference sources). On the other hand, in standard  phased array radars, many small elements are employed so that each element emits an identical signal (up to a phase shift). These phases are shifted to focus the transmit beam in a certain direction \cite{Antenna_theory_and_design}. 
The advantage of transmitting independent waveforms in MIMO radar provides extra degrees of freedom and the waveforms may be optimized across antennas to enhance the performance of radar systems. 

A central problem in MIMO radar is to design a set of waveforms such that the transmitted beampattern matches certain specifications, e.g.\ a desired beampattern, either for narrowband \cite{fuhrmann2008transmit,lipor2014fourier,sajid2014,cheng2017constant,aubry2016mimo} or wideband \cite{san2005beampattern,he2011wideband,yang2012fast,gong2014transmit,aldayel2017tractable} setups. Although the transmit beampattern is used to focus the transmitted power
	in a certain directions of interest, a well-designed beampattern also helps enhance the Signal-to-Noise-Ratio (SNR)\cite{fuhrmann2008transmit,fuhrmann2004transmit,li2010transmit,li2009mimo}. While 	unconstrained design is straightforward, but this is a highly challenging problem in the presence of practical constraints on the waveform \cite{Patton_Thesis}. 

A principally important constraint on the transmit waveforms is the constant modulus constraint (CMC). The  CMC is crucial in the design process due to the presence of non-linear amplifiers in radar systems \cite{patton2008modulus} which must operate in saturation mode. Existing approaches that deal with beampattern design under CMC can be classified into two categories: indirect and direct approaches. The first category consists of methods that approximate or relax the CMC, i.e., the design process is conducted under approximated constraints and then the produced solution is converted to the constant modulus set. Examples in this category include: the peak-to-average ratio (PAR) waveform constraint \cite{stoica2007probing}, \cite{he2011wideband} and the energy constraint \cite{aubry2014radar}. 
In general, since these approaches deal with an approximation of CMC, performance may degrade considerably in an actual real-world scenario when the constraint is strict \cite{Patton_Thesis}. In the second category are methods that directly enforce CMC and hence lead to better performance compared to the indirect approaches, but with relatively expensive computational procedures, such as a quasi Newton iterative method in \cite{wang2012design}, Semidefinite Relaxation (SDR) with randomization \cite{luo2010semidefinite,cui2014mimo} and the sequence of convex programs approach in \cite{aldayel2017tractable}.

 In this work, our goal is to break the trade-off between  performance measures such as faithfulness to the desired beampattern and computational complexity of designing/optimizing the waveform code. We show that this is possible by invoking principles of optimization over manifolds. That is, we derive a new projection, descent, and retraction (PDR) based numerical algorithm that in each iteration allows for monotonic cost function improvement while maintaining feasibility over the complex circle manifold (constant modulus set). Optimization over such a manifold has been investigated for problems in communications for example \cite{chen2018manifold}. However, the work in \cite{chen2018manifold} deals with an entirely different physical problem (and hence cost function) from ours and does not investigate analytical properties of the solution or convergence of the algorithm, which is a key focus of our work for quadratic loss functions.
 
Besides CMC, imposing orthogonality across antennas has been shown to be particularly meritorious. Orthogonal MIMO waveforms enable the radar system to achieve an increased virtual array \cite{bekkerman2006target,li2009mimo} and, hence leads to many practical benefits \cite{fishler2006spatial,li2007parameter,xu2008target}. From a beampattern design standpoint, a compelling practical challenge is that the {``}directional knowledge" of target and interference sources utilized in specifying the desired beampattern may not be perfect. In such scenarios, it has been shown in \cite{bekkerman2006target,li2009mimo} that the gain loss in the  transmit-receive patterns for orthogonal waveform transmission is very small under target direction mismatch.

Some work has been done towards the joint incorporation of CMC and orthogonality constraints \cite{he2009designing,song2015optimization,cui2017constant,li2018fast} under the umbrella of waveforms with desired auto-and-cross correlation properties. In these works, however, the beampattern is {\em not designed} but an outcome. Recently, in \cite{deng2016mimo}, a MIMO beampattern design that incorporates both CMC and orthogonality  was investigated via a numerical approach based on the simulated annealing algorithm. To incorporate CMC, phase of the waveform vector is optimized numerically but analytical properties/guarantees of the solution are not investigated. 

 Our work addresses the aforementioned challenges in transmit MIMO beampattern design focusing particularly on tractable and scalable approaches in the presence of CMC. Specifically, our contributions include:
\begin{itemize}
	\item \textbf{Projection, Descent and Retraction algorithm (PDR):} A new approach is developed that works directly on the complex circle, i.e.\ descent is achieved while maintaining feasibility in the sense of CMC. The proposed numerical update consists of three steps: (1) \textbf{P}rojection of the gradient of the cost function onto the tangent space of the complex circle manifold, i.e. the CMC set, (2) \textbf{D}escent on the tangent space (affine set), and (3) \textbf{R}etraction back to the complex circle. 
	\item \textbf{Algorithm analysis and convergence guarantees:} For quadratic cost functions, we formally prove that the cost function is monotonically decreasing while updating from one point to another on the complex circle and further convergence is guaranteed to a local minima.
	\item \textbf{Incorporating orthogonality:} We show that the aforementioned PDR technique can be applied to enforce orthogonality across antennas by introducing an orthogonality encouraging penalty term with the cost function.
	\item \textbf{Numerical simulations insights and validation:}  We compare the PDR algorithm against the state-of-the-art approaches in MIMO beampattern design which address the CMC constraint. Results show that we can achieve better fidelity against a desired beampattern, at a remarkably lower computational cost. We also show that when orthogonality is incorporated, the PDR algorithm can achieve a beampattern design that exhibits robustness to target direction mismatch, which is hugely desirable in real-world scenarios where the specification of an idealized beampattern may not be exact.
	
\end{itemize}

The rest of the paper is organised as follows. The problem formulation is presented in Section \ref{Sec:System_Model}. In Section \ref{Sec:PDR}, we provide a brief background on optimization over manifolds and develop the proposed PDR algorithm for beampattern design in the presence of the constant modulus constraint, equivalent to optimization over the complex circle manifold. New analytical results are provided in this setting that prove monotonic cost function improvement as well as convergence. Also presented in this section an extension towards incorporating the orthogonality constraint.
Section \ref{Sec:Experiments} compares and contrasts the proposed PDR algorithm against the state-of-the-art approaches in wideband transmit MIMO beampattern design. Concluding remarks and possible future research directions are discussed in Section \ref{Sec:Conclusion}.

	\textbf{Notation:} $\mathbb{C}^N$ and $\mathbb{R}^N$ denote the $N$-dimensional complex and real vector spaces, respectively. We use boldface upper case for matrices and boldface lower case for vectors. $\|\mb x\|_2$ is the 2-norm of the vector $\mb x$ and $\|\mathbf{X}\|_F$ is the Frobenius norm of the matrix $\mb X$. The transpose, the conjugate, and the conjugate transpose (Hermitian) operators are denoted by $(.)^T$, $(.)^*$, and $(.)^H$ respectively. $\odot$ and $\otimes$ denote the Hadamard and Kronecker product respectively. $\text{Re}(.)$ and $\text{Im}(.)$ denote extraction of the real part, and imaginary part of a complex number (or vector), respectively. $|x|$ denotes modulus of the complex number $x$ and $|\mb x|$ is a vector of element wise absolute values of $\mb x$, i.e., $|\mb x|=\big[|x_1|\ |x_2| \ \hdots \ |x_L|\big]^T$. $\nabla_{\mb s}(f(\mb s))$ denotes the gradient of the function $f$ w.r.t. the vector $\mb s$. $\text{vec}(.)$ denotes the vectorization operator. $\mb I_{L}$ is an $L \times L$ identity matrix. For matrices $\mb A$ and $\mb B$, $\mb A \geq \mb B \Rightarrow \mb A-\mb B \geq 0$, i.e., the matrix $\mb A-\mb B$ is positive-semi definite. $\mathbbm{1}_M$ denotes an $M\times M$ matrix with all ones. $\text{ddiag}(\mb A)$ sets all off-diagonal entries of the matrix $\mb A$ to zero.
\begin{figure}[!ht]
	\centering
	\includegraphics[width=1\linewidth]{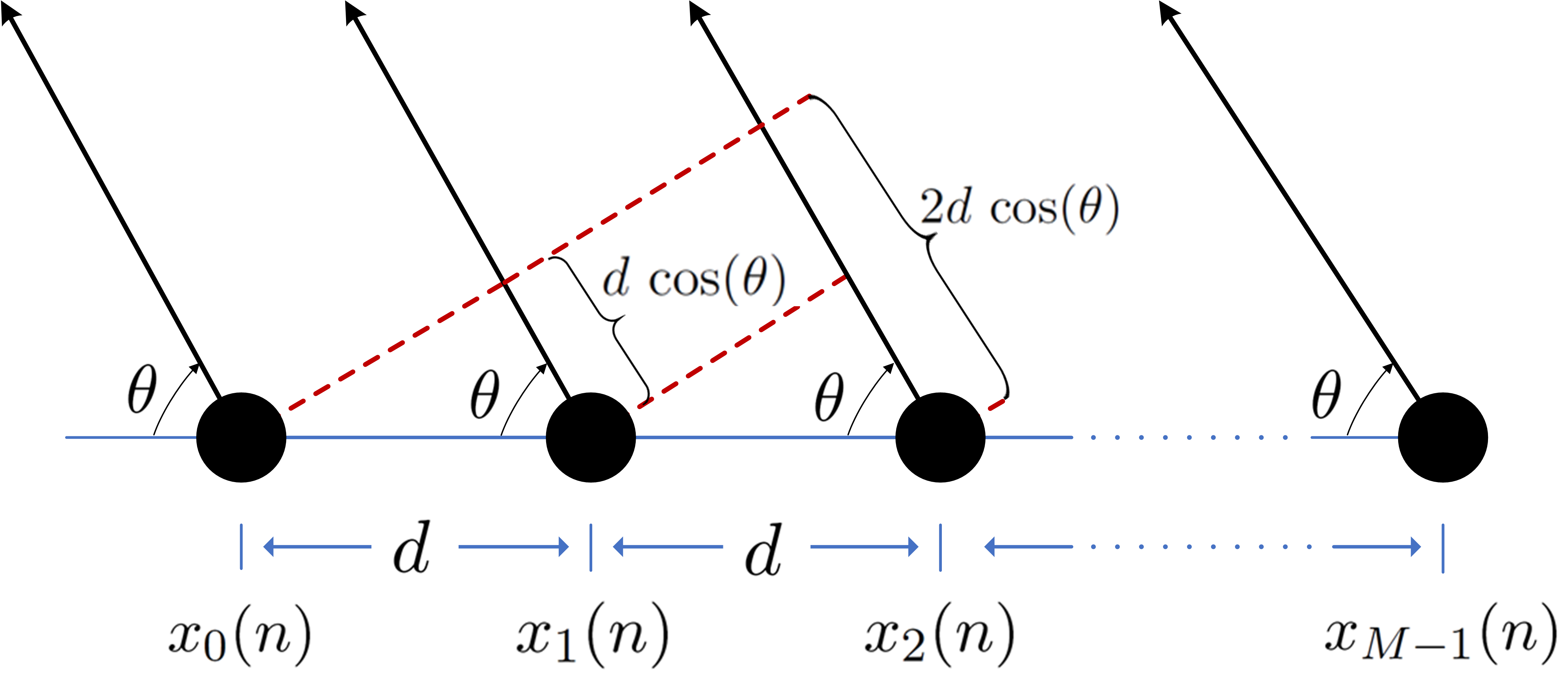}
	\caption{Uniform Linear Array (ULA) MIMO radar system \label{ULA}}
\end{figure}
\section{Problem formulation}
\label{Sec:System_Model}
Consider a Uniform Linear Array (ULA) MIMO radar system that employs  $M$ transmit antennas with inter-element spacing $d$, as shown in Fig. \ref{ULA}. The transmitted bandpass signal transmitted by the $m^{\text{th}}$ antenna is given by 
\begin{IEEEeqnarray}{rCl}  \label{eq:bandpass}
s_m(t) = x_m(t) \ e ^{j2\pi f_c t}
\end{IEEEeqnarray} where $f_c$ is the carrier frequency and  $x_m(t)$ is the baseband signal. 

The baseband signal $x_m(t)$ is sampled to $N$ samples with sampling rate $T_s=1/B$ and the samples are collected in the following vector   
\be \label{eq:disbaseband} \mb x_m = [x_m(0)\quad x_m(1) \quad ... \quad x_m(N-1)]^T \ee where $x_m(n)\triangleq x_m(t=nT_s),\ n = 0,1,\hdots , N-1$, and $B$ is the bandwidth in Hz. 

Let $y_m(p)$ be the discrete Fourier Transform (DFT) of $x_m(n)$ and expressed as:
\be \label{eq:DFT}
y_m(p)=\sum_{n=0}^{N-1} x_m(n) e^{-j2\pi \frac{np}{N}}, \quad p= -\frac{N}{2},...,0, ..., \frac{N}{2}-1
\ee
The discrete frequency beampattern in the far-field at spatial angle $\theta \in [0^{\circ},180^{\circ}]$ is \cite{he2011wideband}:
\be \label{beam1}
P(\theta, p)=|\mb a^H(\theta,p)\mb y_p|^2
\ee
where $\mb a(\theta, p)\in \mathbb{C}^{M\times 1}$ is the steering vector at frequency $\frac{p}{NT_s}+f_c$ defined as
\begin{IEEEeqnarray}{rCl}
	\IEEEeqnarraymulticol{3}{l}{\mb a(\theta, p)=}\nonumber\\* \quad
	& &\big[1\quad e^{j2\pi(\frac{p}{NT_s}+f_c)\frac{d\cos\theta}{c}} \, ...\, e^{j2\pi(\frac{p}{NT_s}+f_c)\frac{(M-1)d\cos\theta}{c}} \big]^T\label{steeringvector}
\end{IEEEeqnarray} 

and $\mb y_p=[y_0(p) \quad y_1(p) \quad ... \quad y_{M-1}(p)]^T$.

Furthermore, the spatial angle $\theta$ can be discretized by dividing the interval $ [0^{\circ},180^{\circ}]$ into $S$ sub-intervals, i.e., $\{\theta_s\}_{s=1}^S$ and hence the steering vector $\mb a(\theta, p)$ can be written in terms of $\theta_s$ and $p$ as $$\mb a_{sp}=\mb a(\theta_s, p),\ s=1,...,S$$ Using these notations, the beampattern in Eq (\ref{beam1}) can be expressed in discrete angle-frequency as $$P_{sp}=|\mb a^H_{sp} \mb y_p|^2=|\mb a^H_{sp} \mb F_p \mb x|^2$$ where $\mb x \in \mathbb{C}^{MN\times 1}$ is the concatenation of the waveforms vectors defined in Eq (\ref{eq:disbaseband}), i.e., \be \label{Eq:vecx}\mb x=[\mb x^T_0 \quad \mb x^T_1 \quad ... \quad \mb x^T_{M-1}]^T\ee and $\mb F_p= \mb e_p\otimes \mb I_M$ 
where $\mb e_p=[1 \quad e^{-j2\pi \frac{p}{N}} \quad ...\quad  e^{-j2\pi \frac{(N-1)p}{N}}]$. The beampattern design problem under constant modulus constraint can be formulated as
\be
\label{Eq:BP}
\begin{array}{cc}
	\displaystyle
	\min_{\mb x} &\sum_{s=1}^{S}\sum_{p=-\frac{N}{2}}^{\frac{N}{2}-1} [d_{sp}-|\mb a^H_{sp} \mb F_p \mb x|]^2\\
	\text{s.t.:  } & |\mb x|= \mb 1

\end{array}
\ee 
where $d_{sp} \in \mathbb{R}$ is the desired beampattern and the constant
modulus constraint ($|\mb x| = \mb 1$) implies that $|x_m(n)| = 1$ for
$m = 0, 1, \hdots ,M$ and $n = 0, 1, \hdots ,N$. We note that the a cost function in Eq (\ref{Eq:BP}) has been the focus of much past work \cite{he2011wideband,yang2012fast,gong2014transmit,aldayel2017tractable} with different ways of specifying $d_{sp}$.

Note that the cost function in Eq (\ref{Eq:BP}) is not tractable due {to} the $|.|$ (absolute value) operator, this operator makes the function non-differentiable w.r.t. the variable $\mb x$.
To overcome this issue, it has been shown in \cite{he2011wideband} that for a generic term $[d_{sp}-|\mb a^H_{sp} \mb F_p \mb x|]^2$ of Eq (\ref{Eq:BP}), the following holds 
\be
\begin{split}
	&\min_{\mb \phi_{sp}} |d_{sp} e^{j\phi_{sp}}- \mb a^H_{sp} \mb F_p \mb x|^2\\&=\min_{\mb \phi_{sp}} \Big\{d_{sp}^2+ |\mb a^H_{sp} \mb F_p \mb x|^2\\ &\ \ \ -2\text{Re}[d_{sp}|\mb a^H_{sp} \mb F_p \mb x|\text{cos}(\phi_{sp}-\text{arg}(\mb a^H_{sp} \mb F_p \mb x))]\Big\}\\&
	= [d_{sp}- |\mb a^H_{sp} \mb F_p \mb x|]^2\ \  (\text{for}\ \phi_{sp}=\arg\{\mb a^H_{sp} \mb F_p \mb x\})
\end{split}
\ee
In view of the above, the authors in \cite{he2011wideband} formulate the following problem over $\phi_{sp}$ and $\mb x$ jointly 
\be
\label{Eq:NA}
\begin{array}{cc}
	\displaystyle
	\min_{\mb x,\{\phi_{sp}\}_{\forall s,p}} &f(\mb x)=\sum_{s=1}^{S}\sum_{p=-\frac{N}{2}}^{\frac{N}{2}-1}  |d_{sp}e^{j\phi_{sp}}-\mb a^H_{sp} \mb F_p \mb x|^2\\
	\text{s.t.:  } & |\mb x|= \mb 1
\end{array}
\ee 
With this form, the beampattern design problem will be carried out over two minimization stages: one w.r.t. $\phi_{sp}$ for fixed $\mb x$ with a minimizer $\phi_{sp}=\arg\{\mb a^H_{sp} \mb F_p \mb x\}$ and the second one will be w.r.t. $\mb x$ for fixed $\phi_{sp}$. It is worthwhile observing that the phase variable $\phi_{sp}$ is not inherent to the problem but introduced to make the problem tractable, i.e., when the phases of $d_{sp}e^{j\phi_{sp}}$ and $\mb a^H_{sp} \mb F_p \mb x$ agree, the cost function is a quadratic w.r.t. $\mb x$. Therefore, in this work, we are seeking to optimize $f(\mb x)$ w.r.t. $\mb x$ for fixed $\phi_{sp}$ under CMC.
 
Then, the cost function in problem (\ref{Eq:NA}) can be rewritten compactly, with fixed $\phi_{sp}$, as:  

\begin{IEEEeqnarray}{rCl} 
	f(\mb x)&=& \sum_{p} \|\mb d_p-\mb A_p \mb F_p \mb x \|_2^2\nonumber\\
	&=&\ \sum_{p} \mb x^H \mb F^H_p \mb A^H_p\mb A_p \mb F_p\mb x - \mb d^H_p \mb A_p \mb F_p\mb x - \mb x^H \mb F^H_p \mb A^H_p \mb d_p\nonumber\\
	&&+ \>\sum_{p}\mb d_p^H\mb d_p\nonumber\\
	&=&\mb x^H \mb P \mb x - \mb q^H \mb x - \mb x^H \mb q+r\label{Eq:cost}
\end{IEEEeqnarray}

where 
\ben
\mb A_{p}={\begin{bmatrix}
		\mb a^H_{1p} \\
		\vdots\\
		\mb a^H_{Sp}
\end{bmatrix}}, \quad
\mb d_{p}={\begin{bmatrix}
		d_{1p} e^{j\phi_{1p}}\\
		\vdots\\
		d_{Sp} e^{j\phi_{Sp}}
\end{bmatrix}},
\een and $\mb P=\sum_{p}  \mb F^H_p \mb A^H_p\mb A_p \mb F_p$, $\mb q=\sum_{p}\mb F^H_p \mb A^H_p \mb d_p$ and $r=\sum_{p}\mb d_p^H\mb d_p$.
Consequently, the problem in (\ref{Eq:NA}) is reframed as 
\be
\label{Eq:P}
\begin{array}{cc}
	\displaystyle
	\min_{\mb x} &f(\mb x)=\mb x^H \mb P\mb x - \mb q^H \mb x - \mb x^H \mb q+r\\
	\text{s.t.:  } & |\mb x|= \mb 1
\end{array}
\ee 
This optimization problem is known to be a hard non-convex NP-hard problem \cite{soltanalian2014designing}. Some of the best known methods that develop a solution for this form are: Wideband Beampattern
Formation via Iterative Techniques (WBFIT) \cite{he2011wideband}, Semi-Definite relaxation with Randomization (SDR) \cite{luo2010semidefinite,cui2014mimo}, 
the Monotonically Error-bound Improving Technique (MERIT)
\cite{soltanalian2014designing}, the Iterative Algorithm for Continuous Phase Case (IA-CPC) \cite{cui2017quadratic}, design algorithm based on Alternating
	Direction Method of Multipliers (ADMM) \cite{Liang16}.
Some of these approaches suffer from low performance accuracy (in terms of deviation from the desired beampattern) while others involve
relatively expensive computational procedures. Importantly, CMC is extracted in different parts of the optimization or in some other methods approached asymptotically \cite{aldayel2017tractable}, but a direct optimization over the non-convex CMC remains elusive.

We take a drastically different approach by invoking principles of optimization over non-convex manifolds. Our focus is on developing a gradient based method, which can enable descent on the complex circle manifold (formal manifold terminology for the CMC) while maintaining feasibility.

Before describing our solution in the next Section, we make an alteration to the cost function by adding $\gamma \mb x^H \mb x$: 
\be \label{Eq:P2}
\begin{array}{cc}
	\displaystyle
	\min_{\mathbf{x} \in \mathbb{C}^{L}} &	\bar{f}(\mb x)=\mb x^H (\mb P+\gamma \mb I)\mb x - \mb q^H \mb x - \mb x^H \mb q\\
	\text{s.t.:  } & |\mb x|= \mb 1\\
\end{array}
\ee
where $L=MN$, and $\gamma \geq0$ (it will be used later in Lemma \ref{Lemma:L2} to control convergence). Since the problem in (\ref{Eq:P2}) enforces CMC, the term $\gamma \mb x^H \mb x$ is constant (i.e. $\gamma \mb x^H \mb x = \gamma L$). Hence, the optimal solution of the problem in (\ref{Eq:P}) and the optimal solution of the problem in (\ref{Eq:P2}) are identical for any $\gamma \geq0$.  

\section{Constant modulus constraint and optimization over manifolds}
\label{Sec:PDR}
The search space or the feasible set of the problem in (\ref{Eq:P2}) can be thought of as the product of $L$ (complex) circles, i.e.,
$$\underbrace{\mathcal{S}\times \mathcal{S} \hdots \times \mathcal{S}}_{L \ \text{times}}$$ where $\mathcal{S}$ is one (complex) circle which is defined as $ \mathcal{S} \overset{\triangle}{=} \{ x \in \mathbb{C} : x^* x=\text{Re}\{x\}^2+\text{Im}\{x\}^2 = 1\}$. This set ($\mathcal{S}$) can be seen as a sub-manifold of $\mathbb{C}$ \cite{bandeira2017tightness} and hence, the product of such $L$ circles is a sub-manifold of $\mathbb{C}^L$ \cite{bandeira2017tightness}. This manifold is known as  \emph{the complex circle manifold} and defined as 
\be  \label{Eq:complexmanifol}                                                                   \mathcal{S}^L \overset{\triangle}{=}\{\mb x \in \mathbb{C}^{L}:|x_l|=1, \ l = 1,2,\hdots,L\}
\ee
Before we proceed with the solution of the optimization problem in (\ref{Eq:P2}), we first provide a background
on optimization over manifolds \cite{absil2009optimization} and subsequently develop
a new technique to optimize directly over the complex circle
manifold.
\subsection{Optimization over manifolds}\label{Optimization_over_manifolds}
The term optimization over manifolds refers to a class of problems of the form
\be \label{man1}
\begin{aligned}
	& \underset{\mathbf{x} \in \mathcal{M}}{\text{min}}
	& & g(\mathbf{x})\\
\end{aligned}
\ee 
where $g(\mathbf{x})$ is a smooth real-valued function and $\mathcal{M}$ is some manifold. Many classical line-search algorithms from unconstrained nonlinear optimization in $\mathbb{C}^L$ such as gradient descent can be used in optimization over manifolds but with some modifications. In general, line-search methods in $\mathbb{C}^L$ are based on the following update formula 
\be \label{SD}\mb x_{(k+1)}=\mb x_{(k)} +\beta_k \boldsymbol{\eta}_{(k)} \ee
where $ \boldsymbol{\eta}_{(k)}\in \mathbb{C}^L$ is the \emph{search direction} and $\beta_k\in \mathbb{R}$ is the \emph{step size}. The most obvious choice for the search direction is the steepest descent direction which is the negative gradient of $g(\mathbf{x})$, i.e., $\boldsymbol{\eta}_{(k)}= -\nabla_{\mb x} g(\mb x_{(k)})$. In the literature \cite{absil2009optimization,kovnatsky2016madmm}, the following high level structure is suggested: 
\begin{itemize}
	\item The descent will be performed on the manifold itself rather than in the Euclidean space by means of the \textit{intrinsic} gradient. The \textit{intrinsic} gradient $\nabla_{\mathcal{M}}g(\mb x_{(k)})$ of $g(\mb x)$ at point $\mb x_{(k)}\in \mathcal{M}$ is a vector
	in the tangent space $\mathcal{T}_{\mb x_{(k)}}\mathcal{M}$ (for the definition of $\mathcal{T}_{\mb x_{(k)}}\mathcal{M}$, see \cite{absil2009optimization}, Section 3.5.7). This \textit{intrinsic} gradient can be obtained by projecting the standard (Euclidean) gradient
	$\nabla_{\mb x} g(\mb x_{(k)})$ onto $\mathcal{T}_{\mb x_{(k)}}\mathcal{M}$ by means of a projection operator $\mathbf{P}_{\mathcal{T}_{\mb x_{(k)}}\mathcal{M}}\big(\nabla_{\mb x} g(\mb x_{(k)})\big)$.
	\item The update is performed on the tangent space along the direction of $\mathbf{P}_{\mathcal{T}_{\mb x_{(k)}}\mathcal{M}}\big(\nabla_{\mb x} g(\mb x_{(k)})\big)$ with a step $\beta$, i.e., $\bar{\mb x}_{(k)}=\mb x_{(k)}- \beta \mathbf{P}_{\mathcal{T}_{\mb x_{(k)}}\mathcal{M}}(\nabla_{\mb x} g(\mb x_{(k)}))$ where $\bar{\mb x}_{(k)} \in \mathcal{T}_{\mb x_{(k)}}\mathcal{M}$.
	\item Since $\bar{\mb x}_{(k)} \notin  \mathcal{M}$, it will be mapped back to the manifold by the means of a \textit{retraction} operator, $\mb x_{(k+1)}=\mathbf{R}\big(\bar{\mb x}_{(k)}\big) $.
\end{itemize}
For some manifolds, the projection $\mathbf{P}_{\mathcal{T}_{\mb x_{(k)}}\mathcal{M}}(.)$
and retraction $\mathbf{R}(.)$ operators admit a closed form. Interested readers may refer to \cite{absil2009optimization} for details.
In the following subsections, we develop new results that employ the ideas articulated above in optimization over manifolds to the complex circle manifold (constant modulus set), i.e., $\mathcal{M}=\mathcal{S}^L$. In particular, 1) we derive expressions for the projection and retraction operators for the complex circle manifold $\mathcal{S}^L$, 2) establish new analytical results which include proof of monotonic cost function improvement while maintaining feasibility in $\mathcal{S}^L$, and 3) provide convergence guarantees. 
\subsection{Complex circle manifold}
The projection and retraction operators for the complex circle manifold are inspired by those for \emph{the unit circle manifold} \cite{absil2009optimization}, which is defined as 
\be                                                                     \mathcal{C}^1 \overset{\triangle}{=}\{\mb q \in \mathbb{R}^{2}: \mb q ^T \mb q =1\}
\ee
That is true because each entry of any feasible vector $\mb x \in \mathcal{S}^L$
can be viewed as a point in $\mathbb{R}^2$ with its components on the (real) unit-circle ($\mathcal{C}^1$). Thus, in order to derive these operators for the complex circle manifold we need to start with the operators for the unit circle manifold. 

First, let us define the following two vectors $\mb z \in \mathcal{S}^L$ and $\mb w \in \mathbb{C}^L$, and let $\bar{\mb z}_l$ and $\bar{\mb w}_l$ be 2-dimensional vectors represent the components of the $l^{\text{th}}$ element of $\mb z$ and $\mb w$ respectively, i.e.,
\begin{itemize}
	\item $\bar{\mb z}_l=[\text{Re}\{z_l\} \ \text{Im}\{z_l\}]^T=[z_{lr} \ z_{li}]^T\in \mathbb{R}^2$ 
	\item $\bar{\mb w}_l=[\text{Re}\{w_l\} \ \text{Im}\{w_l\}]^T=[w_{lr} \ w_{li}]^T\in \mathbb{R}^2$ 
\end{itemize}
where $l = 1,2,\hdots, L$. Note that since the vector $\mb z \in \mathcal{S}^L$, then the vector $\bar{\mb z}_l \in \mathcal{C}^1$, i.e., $\bar{\mb z}_l^T\bar{\mb z}_l=1$. 

The projection of $\bar{\mb w}_l \in \mathbb{R}^{2}$ to the tangent space of the unit circle manifold\footnote{The tangent space for unit circle manifold at $\mb q \in \mathcal{C}^1$ is defined as  $\mathcal{T}_{\mb q}\mathcal{C}^1=\{\mb h \in \mathbb{R}^{2}: \mb q^T \mb h =0\}$\cite{absil2009optimization}.} $\mathcal{T}_{\bar{\mb z}_l}\mathcal{C}^1$ at  $\bar{\mb z}_l$ is given by \cite{absil2009optimization}:
\begin{IEEEeqnarray}{rCl} 
	\text{Proj}_{\mathcal{T}_{\bar{\mb z}_l}\mathcal{C}^1}(\bar{\mb w}_l)
	& = & \bar{\mb w}_l-\bar{\mb z}_l^T \bar{\mb w}_l \bar{\mb z}_l\nonumber\\
	& = & \bar{\mb w}_l-(w_{lr}z_{lr}+w_{li}z_{li})\bar{\mb z}_l\nonumber\\
	& = & \bar{\mb w}_l-\text{Re}\{w_l^{*}z_l\}\bar{\mb z}_l\nonumber\\
	& = &{\begin{bmatrix}
			w_{lr}-\text{Re}\{w_l^{*}z_l\}z_{lr} \\
			w_{li}-\text{Re}\{w_l^{*}z_l\}z_{li} \\
	\end{bmatrix}}\in \mathbb{R}^2 \label{unitcircleproj}
\end{IEEEeqnarray} 
and the retraction of any $\bar{\mb w}_l \in \mathbb{R}^{2}$ to $\mathcal{C}^1$ is given by \cite{absil2009optimization} 
\begin{IEEEeqnarray}{rCl} 
	\text{Ret}(\bar{\mb w}_l)&=& \frac{\bar{\mb w}_l}{\|\bar{\mb w}_l\|_2}\nonumber\\
	&=& 
\Bigg[\frac{w_{lr}}{\sqrt{w_{lr}^2+w_{li}^2}}
\ \frac{w_{li}}{\sqrt{w_{lr}^2+w_{li}^2}}\Bigg]^T\in \mathcal{C}^1\label{unitcircleretraction}
\end{IEEEeqnarray}
These operators for projecting $\bar{\mb w}_l$ onto $\mathcal{T}_{\bar{\mb z}_l}\mathcal{C}^1$ and retracting it to $\mathcal{C}^1$  are illustrated in Fig. \ref{unit_circle_pro_Ret}.
\begin{figure}[!ht]
	\centering
	\includegraphics[scale=0.15]{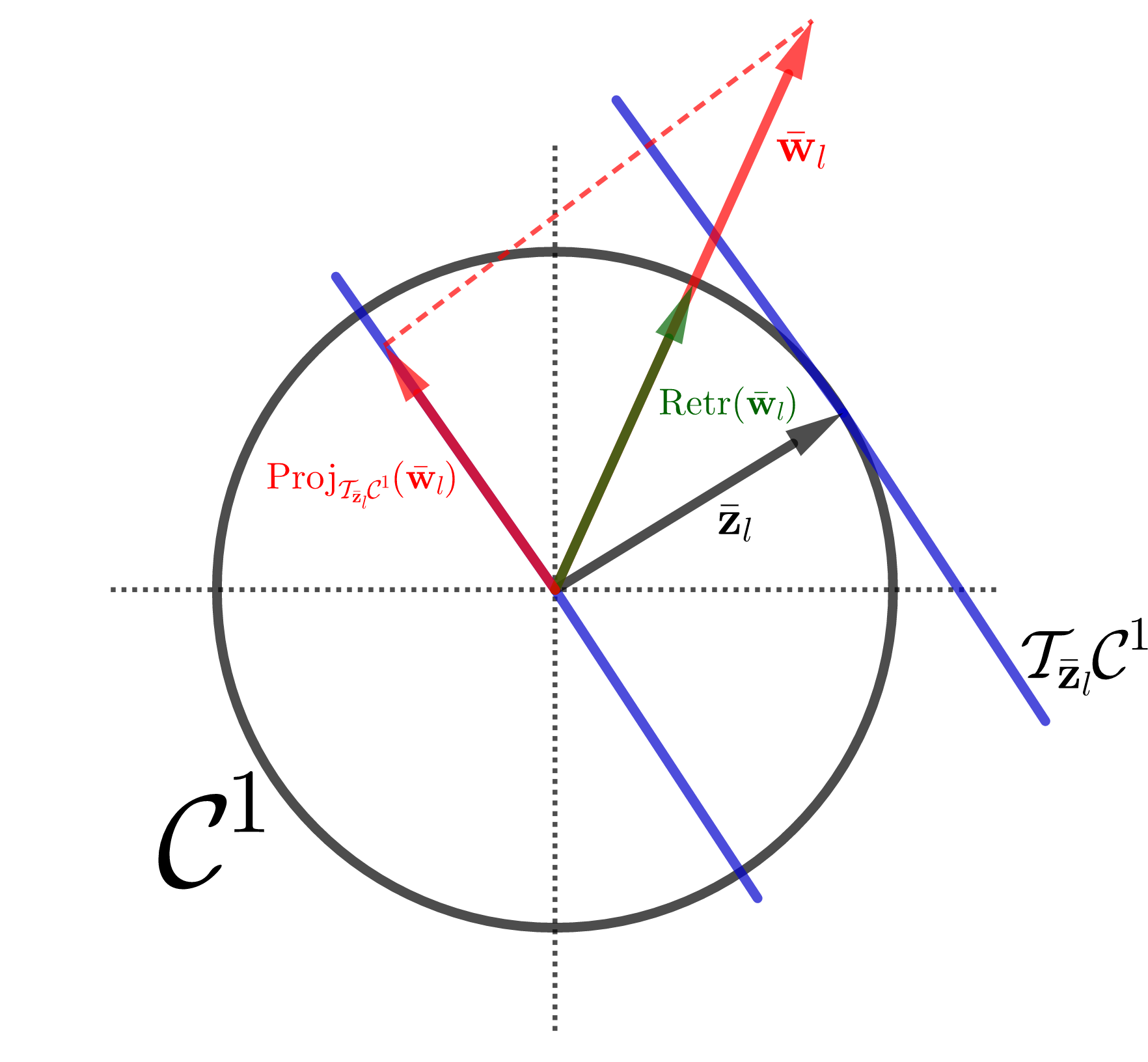}
	\caption{Projection $\text{Proj}_{\bar{\mb z}_l}(\bar{\mb w}_l)$ and retraction $\text{Retr}(\bar{\mb w}_l)$ operators in $\mathcal{C}^1$.}
	\label{unit_circle_pro_Ret}
\end{figure}

According to our discussion at the beginning of this subsection, the relation between the projection of $w_l$ onto the tangent space\footnote{The tangent space of $\mathcal{S}$ at the point $z_l \in \mathcal{S}$ is defined as \cite{bandeira2017tightness} $\mathcal{T}_{z_l}\mathcal{S} = \{y\in \mathbb{C}: \text{Re}\{y^*  z_l\}=0\}$} $\mathcal{T}_{z_l}\mathcal{S}$ and the projection of $\bar{\mb w}_l$ onto the tangent space  $\mathcal{T}_{\bar{\mb z}_l}\mathcal{C}^1$ is similar to the relation between $w_l$ and $\bar{\mb w}_l$ which is $\bar{\mb w}_l=[\text{Re}\{w_l\} \ \text{Im}\{w_l\}]^T$. Given this relation, the projection defined in {Eq (\ref{unitcircleproj})} can be used to project $\mb w$ element-wise onto the tangent space\footnote{This tangent space is the product of $L$ tangent spaces of those for the manifold $\mathcal{S}$, i.e., $\mathcal{T}_{\mb z}\mathcal{S}^L=\mathcal{T}_{z_1}\mathcal{S}\times \mathcal{T}_{z_2}\mathcal{S}\hdots \times \mathcal{T}_{z_L}\mathcal{S}$.} $\mathcal{T}_{\mb z}\mathcal{S}^L$, denoted as $\mathbf{P}_{\mathcal{T}_{\mb z}\mathcal{S}^L}(\mb w)$. Rearranging the entries of the 2-dimensional vector in Eq (\ref{unitcircleproj}) as real and imaginary components for the $l^{\text{th}}$ entry  of  $\mathbf{P}_{\mathcal{T}_{\mb z}\mathcal{S}^L}(\mb w)$ yields
\begin{IEEEeqnarray}{rCl}
	\IEEEeqnarraymulticol{3}{l}{\mathbf{P}_{\mathcal{T}_{\mb z}\mathcal{S}^L}(\mb w)}\nonumber\\* \quad
	&=& {\begin{bmatrix}
			w_{1r}-\text{Re}\{w_1^{*}z_1\}z_{1r}+j(w_{1i}-\text{Re}\{w_1^{*}z_1\}z_{1i}) \nonumber\\
			\vdots\\
			w_{Lr}-\text{Re}\{w_L^{*}z_L\}z_{Lr}+j(w_{Li}-\text{Re}\{w_L^{*}z_L\}z_{Li})
	\end{bmatrix}}\nonumber\\
        &=&{\begin{bmatrix}
        		w_1-\text{Re}\{w_1^{*}z_1\}z_1 \\
        		\vdots\\
        		w_L-\text{Re}\{w_N^{*}z_L\}z_L
     \end{bmatrix}}=\mb w - \text{Re}\{\mb w^{*}\odot\mb z\}\odot \mb z \label{compcir2}
\end{IEEEeqnarray}
Similarly, for the retraction operator in $\mathbb{C}^L$, rearranging the entries of the vector in Eq (\ref{unitcircleretraction}) for each entry  of  $\mathbf{R}(\mb w)$ yields
\begin{IEEEeqnarray}{rCl}
	\mathbf{R}(\mb w)&=&{\begin{bmatrix}
			\frac{w_{1r}}{\sqrt{w_{1r}^2+w_{1i}^2}}+j\frac{w_{1i}}{\sqrt{w_{1r}^2+w_{1i}^2}}\\
			\vdots\\
			\frac{w_{Lr}}{\sqrt{w_{Lr}^2+w_{Li}^2}}+j\frac{w_{Li}}{\sqrt{w_{Lr}^2+w_{Li}^2}}
	\end{bmatrix}}\nonumber\\
&=&{\begin{bmatrix}
		\frac{w_1}{|w_1|}\\
		\vdots\\
		\frac{w_L}{|w_L|}
\end{bmatrix}}=\mb w \odot \frac{1}{|\mb w|}\label{compcir4}
\end{IEEEeqnarray}
\subsection{Projection, Descent and Retraction (PDR) algorithm}\label{PDRsection}
Given the projection in Eq (\ref{compcir2}) and the retraction in Eq (\ref{compcir4}),  the optimization steps over manifolds (described in the subsection \ref{Optimization_over_manifolds}) can be computed to solve the problem in (\ref{man1}) over the complex circle manifold ($\mathcal{M}=\mathcal{S}^L$). 
Precisely, this problem can be solved iteratively by preforming the following steps at each iteration $k$: \begin{enumerate}
	\item A \textbf{P}rojection of the search direction $\boldsymbol{\eta}_{(k)}=-\nabla_{\mb x} g(\mb x_{(k)})$ onto the tangent space $\mathcal{T}_{\mb x_{(k)}}\mathcal{S}^L$ using  Eq (\ref{compcir2}).
	\item A \textbf{D}escent on this tangent space to update the current value of $\mb x_{(k)}$ on the tangent space $\mathcal{T}_{\mb x_{(k)}}\mathcal{S}^L$
	 as
		\ben \bar{\mathbf{x}}_{(k)}=\mathbf{x}_{(k)} +\beta\mathbf{P}_{\mathcal{T}_{\mb x_{(k)}}\mathcal{S}^L}\big(\boldsymbol{\eta}_{(k)}\big)\een
	\item A \textbf{R}etraction of this update to $\mathcal{S}^L$ by using Eq (\ref{compcir4}) as $$\mathbf{x}_{(k+1)}=\mathbf{R}\big(\bar{\mathbf{x}}_{(k)}\big)$$
\end{enumerate}
 Note that $\mathbf{x}_{(k)}, \mathbf{x}_{(k+1)}$ are both on the complex circle manifold, i.e. CMC points, while $\bar{\mathbf{x}}_{(k)}$ is generally a non-CMC point with magnitude $\geq 1$. The proposed algorithm from these steps is named as  \textbf{P}rojection-\textbf{D}escent-\textbf{R}etraction (\textbf{PDR}) and it is visually illustrated in Fig. \ref{fig:complexcircle}.
\begin{figure}[!ht] 
	\centering
	\includegraphics[scale=0.17]{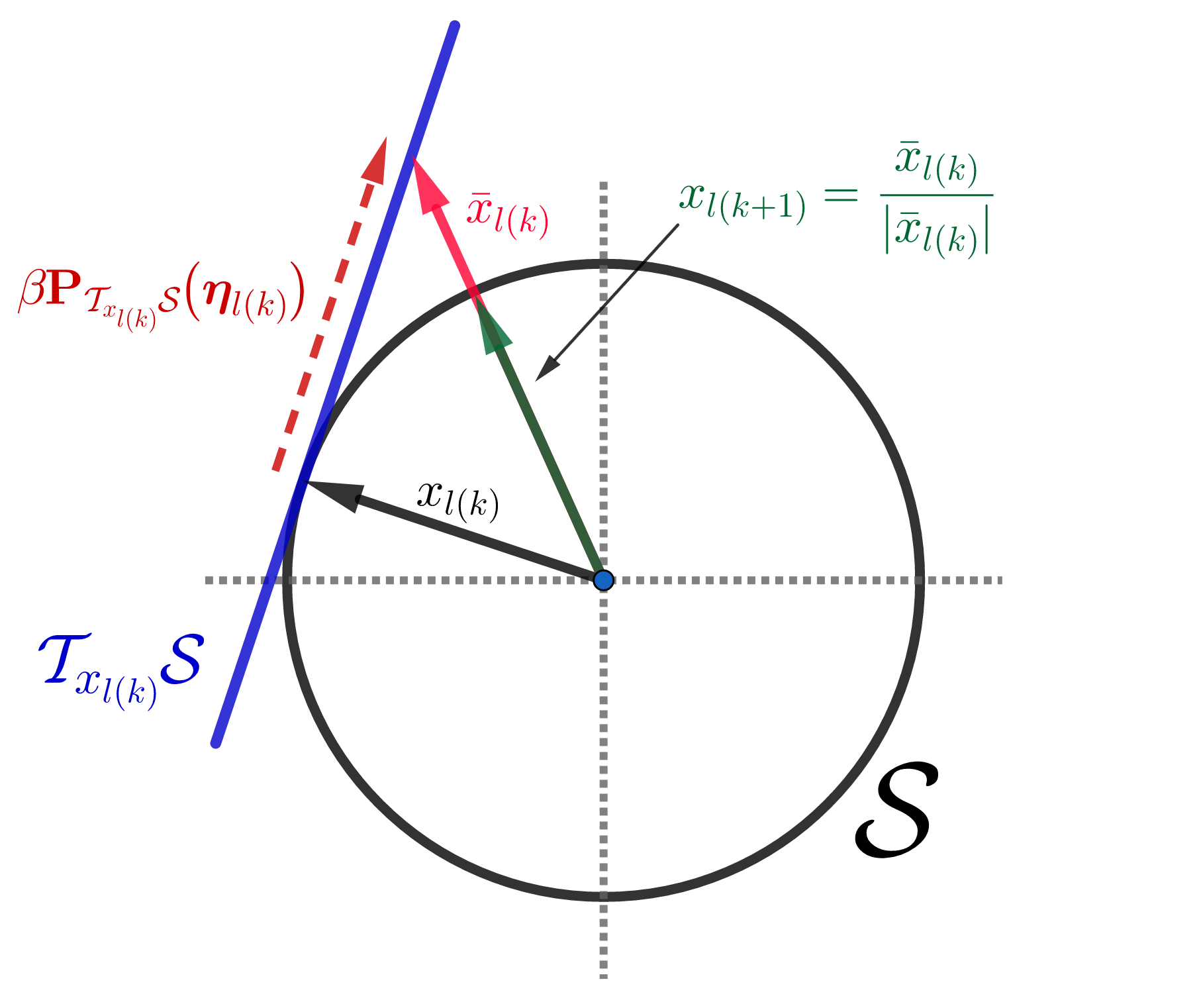}
	\caption{Illustration of the update $x_{l(k+1)}$ starting from $x_{l(k)}$, where $x_{l(k)}$ is the $l$-th element of the complex vector ${\mathbf x}_{(k)}$.}
	\label{fig:complexcircle}                                                                           
\end{figure}

\textbf{PDR for beampattern design:} The cost function of the beampattern design problem defined in (\ref{Eq:P2}) is a quadratic w.r.t. the complex variable $\mb x$ and its  gradient is given by $\nabla_{\mb x} \bar{f}(\mb x) = 2({\mb P} + \gamma {\mb I}) {\mb x} - 2 {\mb q} $. The procedure of minimizing the cost function $\bar{f}(\mb x)$ using the PDR approach is formally described in  Algorithm \ref{Alg:PDR_basic}. 
The convergence of Algorithm \ref{Alg:PDR_basic} is studied in the following sub-section. Assuming that Algorithm \ref{Alg:PDR_basic} converges, the convergence of Algorithm \ref{Alg:PDR} (which alternates between $\phi_{sp}$ and $\mb x$ and achieves practical beampattern design) is guaranteed and already established in past works \cite{sussman1962least,Gerchberg72,he2011wideband,aldayel2017tractable}.

\begin{algorithm}[!ht]
	\caption{Projection-Descent-Retraction (PDR)}
	\label{Alg:PDR_basic}
	\begin{algorithmic}
		\State \tb{Inputs: } The cost function $\bar{f}(\mb x)$, $\mb x_{(0)} \in \mathcal{S}^L$, a step size $\beta$ and a pre-defined threshold value $\epsilon$.
		\State \tb{Output: } A solution $\mb x^{\star}$ for optimizing $\bar{f}(\mb x)$ over the complex circle manifold $\mathcal{S}^L$.
		\State (1) Set $k=0$.
		\State (2) Evaluate the search direction: $\boldsymbol{\eta}_{(k)}=-\nabla_{\mb x} \bar{f}(\mb x_{(k)})$ with computational complexity of $\mathcal{O}(L^2)$.
		\State (3) Compute the projection of the $\boldsymbol{\eta}_{(k)}$ onto the tangent space according to Eq (\ref{compcir2}) as  
		\be \mathbf{P}_{\mathcal{T}_{\mb x_{(k)}}\mathcal{S}^L}\big(\boldsymbol{\eta}_{(k)}\big)=\boldsymbol{\eta}_{(k)} - \text{Re}\{\boldsymbol{\eta}_{(k)}^{*}\odot \mb x_{(k)}\}\odot \mb x_{(k)}
		\ee
		with computational complexity of $\mathcal{O}(L)$.
		\State (4) Compute the update of $\mb x_{(k)}$ on $\mathcal{T}_{\mb x_{(k)}}\mathcal{S}^L$ as
		\be \label{alg1}
		\begin{split}
		 \bar{\mathbf{x}}_{(k)}&=\mathbf{x}_{(k)} +\beta \mathbf{P}_{\mathcal{T}_{\mb x_{(k)}}\mathcal{S}^L}\big(\boldsymbol{\eta}_{(k)}\big)\\
		 &=\mathbf{x}_{(k)} +\beta \big(\boldsymbol{\eta}_{(k)} - \text{Re}\{\boldsymbol{\eta}^*_{(k)}\odot \mb x_{(k)}\}\odot \mb x_{(k)}\Big)
		\end{split}\ee
		\State (5) Compute the next iterate by retracting  $\bar{\mathbf{x}}_{(k)}$ to the complex circle manifold by using the retraction formula Eq (\ref{compcir4}) as
		\be \label{retrac2}\mathbf{x}_{(k+1)}=\mathbf{R}\big(\bar{\mathbf{x}}_{(k)}\big)
		\ee
		with computational complexity of $\mathcal{O}(L)$.
		\If {$|\bar{f}(\mathbf{x}_{(k+1)})-\bar{f}(\mathbf{x}_{(k)})|< \epsilon$}
		\State STOP.
		\Else
		\State k = k+1.
		\State GOTO step (3).
		\EndIf
		\State \tb{Output:} $\mb x^{\star}_{\text{Alg\ref{Alg:PDR_basic}}}= \mb x_{(k+1)}$
	\end{algorithmic}
\end{algorithm}
\textbf{Computational complexity of PDR (Algorithm \ref{Alg:PDR_basic}):} 
As can be inferred from the step-wise description of Algorithm \ref{Alg:PDR_basic}, there are two key steps of complexity $\mathcal{O}(L^2)$ and $\mathcal{O}(L)$ respectively. For large $L$, PDR's complexity per iteration when optimizing quadratic cost functions is dominantly $\mathcal{O}(L^2)$, where $L =NM$.  Table \ref{Complexity} shows computational complexity for the following state-of-the-art approaches along with PDR: Wideband Beampattern Formation via Iterative Techniques (WBFIT) \cite{he2011wideband}, Semi-Definite relaxation with Randomization (SDR) \cite{luo2010semidefinite,cui2014mimo}, Iterative Algorithm for Continuous Phase Case (IA-CPC) \cite{cui2017quadratic} and design algorithm based on Alternating Direction Method of Multipliers (ADMM) \cite{Liang16}, where $F$ is the total number of iterations and $T$ denotes the number of randomization trails for SDR. From this table, it can be seen that  PDR exhibits lower complexity compared to SDR and similar complexity to ADMM and IA-CPC (per iteration). The ADMM and IA-CPC approaches however need more iterations (larger $F$) to achieve the same performance as PDR (as demonstrated in Section \ref{Sec:Experiments}). Note that the complexity of competing methods is reported as derived in their respective/past work.
\begin{table}[H]
	\caption{Computational  complexity for the state of the art algorithms under the constant modulus constraint.} \label{Complexity}
	\centering
	\begin{tabular}{|c|c|}
		\hline
		\textbf{Method}                            & Complexity order                                   \\ \hline
		WBFIT\cite{he2011wideband}                 & $\mathcal{O}(FNM^2)$                               \\ \hline
		ADMM\cite{Liang16}                         & $\mathcal{O}(FN^2M^2)$                             \\ \hline
		SDR\cite{luo2010semidefinite,cui2014mimo}  & $\mathcal{O}(N^{3.5}M^{3.5})+\mathcal{O}(TN^2M^2)$ \\ \hline
		IA-CPC \cite{cui2017quadratic}             & $\mathcal{O}(FN^2M^2)$                             \\ \hline
		PDR                                        & $\mathcal{O}(FN^2M^2)$                             \\ \hline
	\end{tabular}
\end{table}
\begin{algorithm}[!ht]
	\caption{Projection-Descent-Retraction (PDR) for the beampattern design problem}
	\label{Alg:PDR}
	\begin{algorithmic}
		\State \tb{Inputs: }$d_{sp}$, $\mb F_p$, $\mb a_{sp}$ for $p= -\frac{N}{2},...,0, ..., \frac{N}{2}-1$, $s=1, 2, .., S$, $\mb x^{(0)} \in \mathcal{S}^L$, a step size $\beta$ and pre-defined threshold values $\epsilon$ and $\zeta$.
		\State \tb{Output: } A solution $\mb x^{\star}$ for the problem in Eq (\ref{Eq:P}).
		\State (1) Compute $\mb P=\sum_{p}  \mb F^H_p \mb A^H_p\mb A_p \mb F_p$
		\State (2) Set $m=1$.
		\State (3) Set $\phi_{sp}=\text{arg} (\mb a^H_{sp} \mb F_p \mb x^{(m-1)}) \ \forall \ s \ \text{and} \ p$.
		\State (4) Update $\mb d$: $\mb d= [d_{1p} e^{j\phi_{1p}},\hdots , d_{Sp} e^{j\phi_{Sp}}]^T \ \forall \ p$.
		\State (5) Update $\mb q$: $\mb q=\sum_{p}\mb F^H_p \mb A^H_p \mb d_p$.
		\State (6) Use Algorithm \ref{Alg:PDR_basic} with the following inputs: the cost function $\bar{f}(\mb x_{(k)})$ defined in Eq (\ref{Eq:P}), $\mb x_{(0)}=\mb x^{(m-1)}$, $\beta$ and $\epsilon$.	
		\State (7) Set $\mb x^{(m)}= x^{\star}_{\text{Alg\ref{Alg:PDR_basic}}}$.
		\If {$\|\mathbf{x}^{(m)}-\mathbf{x}^{(m-1)}\|< \zeta$}
		\State STOP.
		\Else
		\State m = m+1
		\State GOTO step (3).
		\EndIf
		\State \tb{Output:} $\mb x^{\star}= \mb x^{(m)}$
	\end{algorithmic}
\end{algorithm} 

\subsection{Convergence Analysis}
\label{Subsec:CA}
The update $\bar{\mathbf{x}}_{(k)}=\mathbf{x}_{(k)} +\beta\mathbf{P}_{\mathcal{T}_{\mb x_{(k)}}\mathcal{S}^L}\big(\boldsymbol{\eta}_{(k)}\big)$ in Step 4 of Algorithm \ref{Alg:PDR_basic} will produce a point on the tangent space $\mathcal{T}_{\mb x_{(k)}}\mathcal{S}^L$ with a decrease in the cost if the step size $\beta$ is chosen carefully. A condition on the step size that ensures a decrease in the cost function is provided in the following lemma.
\begin{lemma}
  \label{Lemma:L1}
Let $\lambda_{\mb P+\gamma \mb I}$ denote the largest eigenvalue of the matrix ($\mb P+\gamma \mb I$). If the step size $\beta$ satisfies  \be \label{lemma1} 0<\beta <\frac{1}{\lambda_{\mb P+\gamma \mb I}}\ee then $\bar{f}(\mb x_{(k)})\geq \bar{f}(\bar{\mathbf{x}}_{(k)})$.
  \end{lemma}
\begin{proof}
See subsection \ref{Subsection:AppA} of the Appendix.
\end{proof}
In the next lemma, we show that the cost function $\bar{f}(\mb x)$  is non-increasing through the retraction step given that the positive scalar $\gamma$ satisfies a certain condition.

\begin{lemma}
	\label{Lemma:L2}
Let $\lambda_{\mb P}$ denote the largest eigenvalue of the matrix $\mb P$. If \be \label{lemma2} \gamma \geq \frac{L}{8}\lambda_{\mb P}+\|\mb q\|_2 \ee then $\bar{f}(\bar{\mathbf{x}}_{(k)})\geq \bar{f}(\mathbf{x}_{(k+1)})$.
\end{lemma}
\begin{proof}
See subsection \ref{Subsection:AppB} of the Appendix.
\end{proof}
In the following lemma, we show that the original cost function $f(\mb x)$  defined in (\ref{Eq:P}) is non-increasing and the iterative procedure converges.
\begin{lemma}
	\label{Lemma:L3}
Given $\gamma \geq \frac{L}{8}\lambda_{\mb P}+\|\mb q\|_2$ and $0<\beta < 1/{\lambda_{\mb P+\gamma \mb I}}$ the sequence $\{\bar{f}(\mathbf{x}_{(k)})\}_{k=0}^{\infty}$ generated by Algorithm \ref{Alg:PDR_basic} is non-increasing (from Lemmas \ref{Lemma:L1} and \ref{Lemma:L2}) and hence the sequence $\{f(\mathbf{x}_{(k)})\}_{k=0}^{\infty}$ is also non-increasing. Moreover, since $f(\mathbf{x})\geq 0$ (sum of norms), $\forall \ \mb x$ , it is bounded below and converges to a finite value $f^{*}$.
\end{lemma}
\begin{proof}
	See subsection \ref{Subsection:AppC} of the Appendix.
\end{proof}

\noindent \textbf{Remark:\underline{}} The proposed PDR has essentially enabled a gradient based update while maintaining feasibility on the (non-convex) complex circle manifold with guarantees of monotonic cost function decrease and convergence. We note that the guarantees provided here may not necessarily generalize to other non-convex manifolds. It is indeed structure specific to this problem that enables our construction as shown in Fig. \ref{fig:complexcircle}.
\subsection{Orthogonal waveform design across antennas}
\label{Sec:PDR_Orth}
The feasible set of the optimization problem that incorporating CMC and orthogonality to design the beampattern can be understood as the intersection of the aforementioned complex circle manifold and the complex Stiefel manifold \cite{absil2009optimization,wen2013feasible,jiang2015framework}. Working directly on the intersection of these manifolds is difficult (or impossible) because  the intersection of two manifolds is not always a manifold, and even when it is, it may not be easy to describe. Our strategy to deal with this problem is to  optimize over the complex circle manifold while modifying the beampattern design cost function with the addition of a penalty term that emphasizes orthogonality. Specifically, the cost function in (\ref{Eq:P2}) can be altered by adding the following penalty term $\alpha\|\mathbf{X}^H\mathbf{X}-N\mathbf{I}_M\|^2_F, \ \alpha > 0$, where $\mb X \in \mathbb{C}^{N\times M}$ is the transmit waveform matrix 
$\mb X=[\mb x_0 \quad \mb x_1 \quad ... \quad \mb x_{M-1}]$ and related to the vector $\mb x$ defined in Eq (\ref{Eq:vecx}) through the \emph{vectorization} operator, i.e., $\mb x = \text{vec} (\mb X)$. This way, $\mathbf{X}$ will be ``encouraged" to be orthogonal/unitary, and with this penalty term, the optimization problem that assimilates both constraints can be written as
\be \label{Eq:penalty1}
\begin{array}{cc}
	\displaystyle
	\min_{\mathbf{x} \in \mathbb{C}^{L}} &	\begin{split}h(\mb x) = & \mb x^H \mb P\mb x - \mb q^H \mb x - \mb x^H \mb q+r\\&+\alpha \|\mathbf{X}^H\mathbf{X}-N\mathbf{I}_M\|^2_F \end{split}\\
	\text{s.t.:  } & |\mb x|= \mb 1\\
\end{array}
\ee

The gradient of the cost function $h(\mb x)$ w.r.t. $\mb x$ can be computed by utilizing the relation between the variables $\mb x$ and $\mb X$. The penalty term is a scalar function of the matrix $\mb X$, and the gradient of a scalar function with respect to $\mb X$ and $\mb x$ can be also related through the \emph{vec} operator \cite{fackler2005notes}, i.e.,
 \ben \label{Eq:gradrelation}
 \begin{split} 
 	\nabla_{\mb x} \big(\|\mathbf{X}^H\mathbf{X}-N\mathbf{I}_M\|^2_F\big)=\text{vec}\Big(\nabla_{\mb X} \big(\|\mathbf{X}^H\mathbf{X}-N\mathbf{I}_M\|^2_F\big)\Big)
 \end{split}
 \een
 using this relation, the gradient of the penalty term w.r.t. $\mb x$ is given by
 \ben\label{gradvec2}
 \begin{split} 
 	\nabla_{\mb x} \big(\|\mathbf{X}^H\mathbf{X}-N\mathbf{I}_M\|^2_F\big)&=\text{vec}\Big(\nabla_{\mb X} \big(\|\mathbf{X}^H\mathbf{X}-N\mathbf{I}_M\|^2_F\big)\Big)\\&= 4\  \text{vec}(\mb X \mb X^H \mb X)
 \end{split}
 \een 
 and hence the gradient of the cost function $h(\mb x)$ is given by
 \be \label{Eq:Gradpenalty}
 \begin{split}
 	\nabla_{\mb x} h(\mb x) = 2({\mb P} + \gamma {\mb I}) {\mb x} - 2 {\mb q} + 4 \alpha \text{vec}(\mb X \mb X^H \mb X)
 \end{split}
 \ee
The optimization problem in (\ref{Eq:penalty1}) can now be solved by executing Algorithm \ref{Alg:PDR_basic} but with a different cost function, i.e. $h(\mathbf{x})$. As before, practical beampattern design in the presence of CMC while encouraging orthogonality can be obtained by Algorithm 2, which invokes Algorithm \ref{Alg:PDR_basic} in Step 6.
\begin{figure*}
	\centering
	\subfloat[Unconstrainted]{\includegraphics[width=0.49\linewidth]{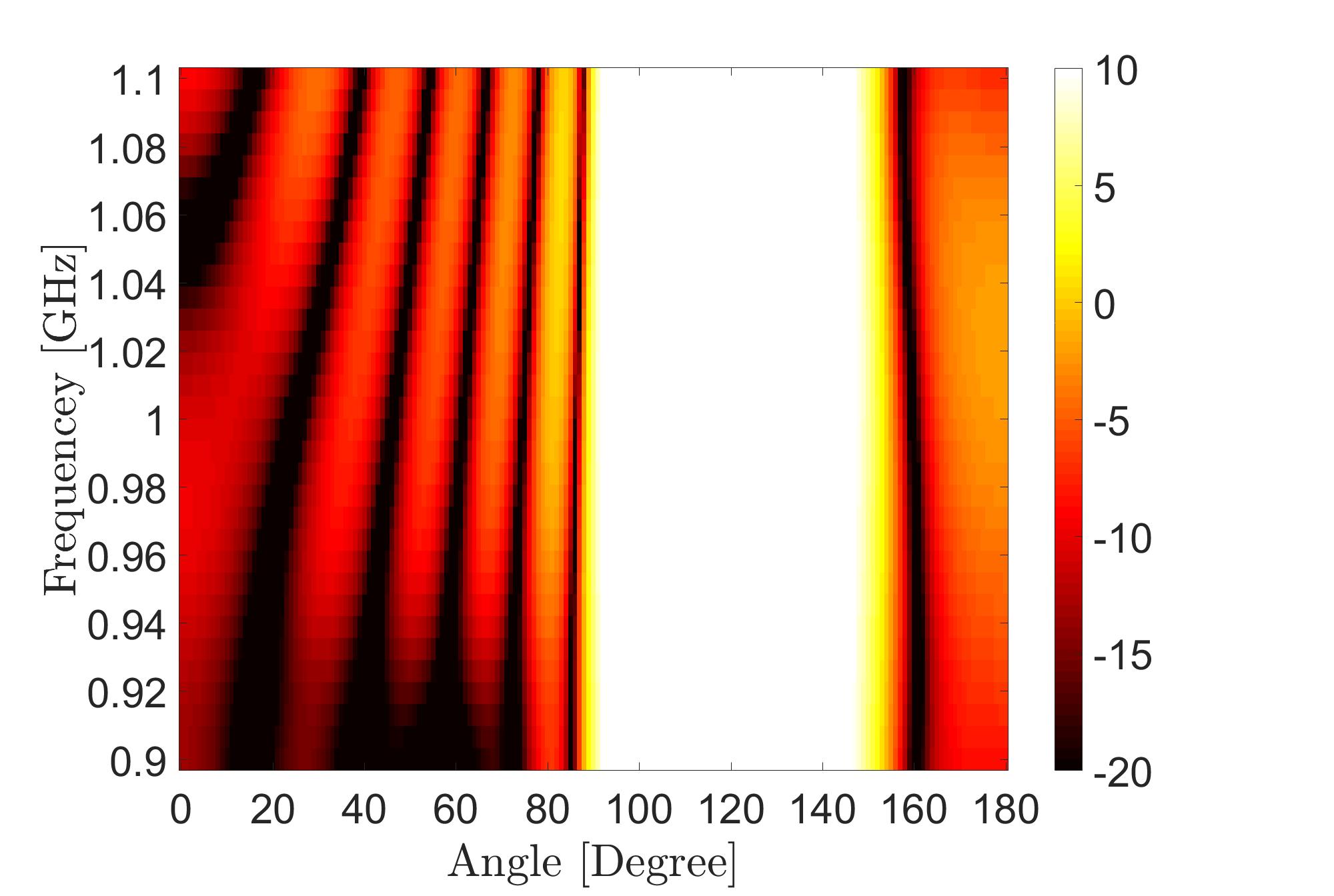}\label{desired1}}\quad
	\subfloat[WBFIT\cite{he2011wideband}]{\includegraphics[width=0.49\linewidth]{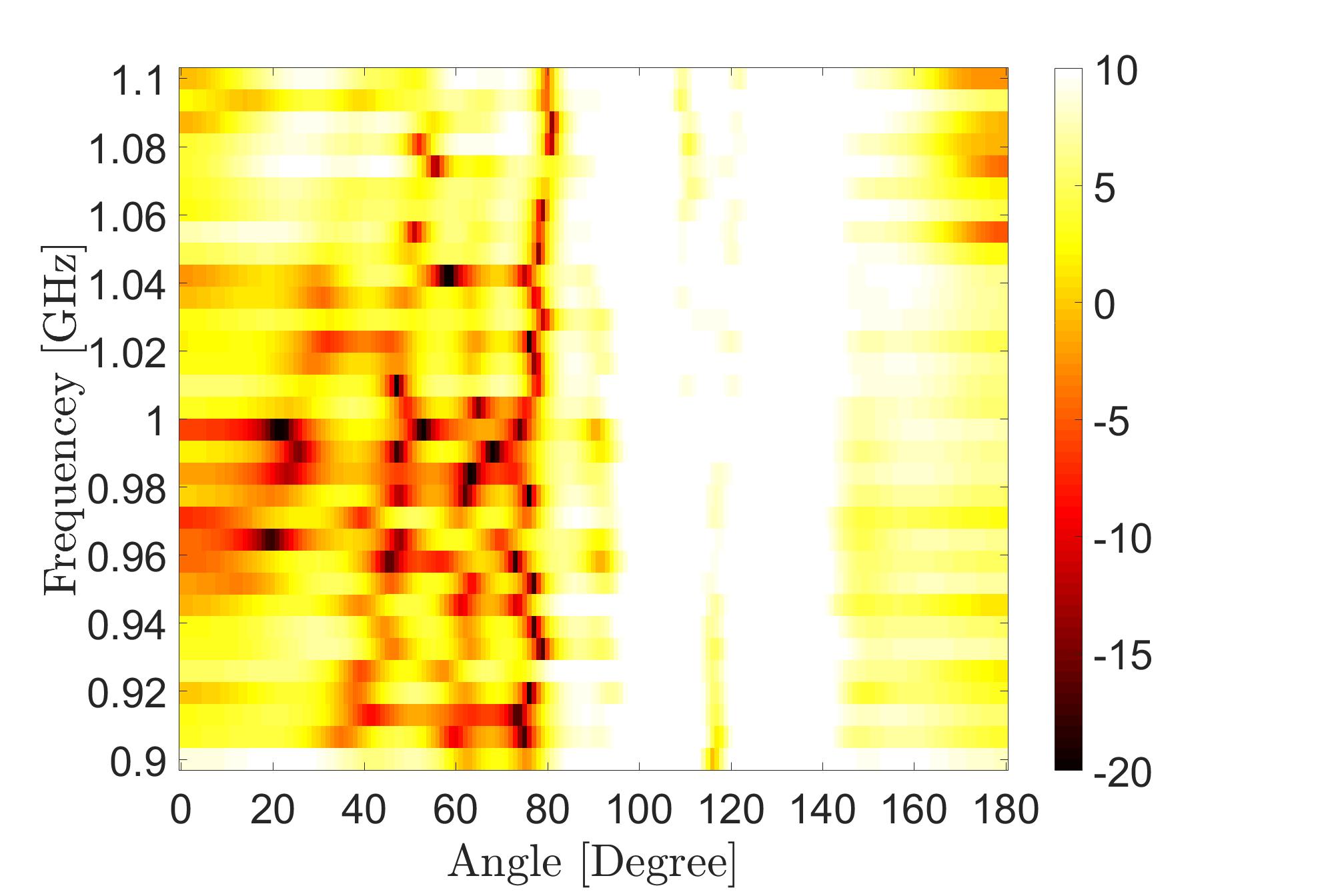}}\\
	\subfloat[SDR\cite{luo2010semidefinite,cui2014mimo}]{\includegraphics[width=0.49\linewidth]{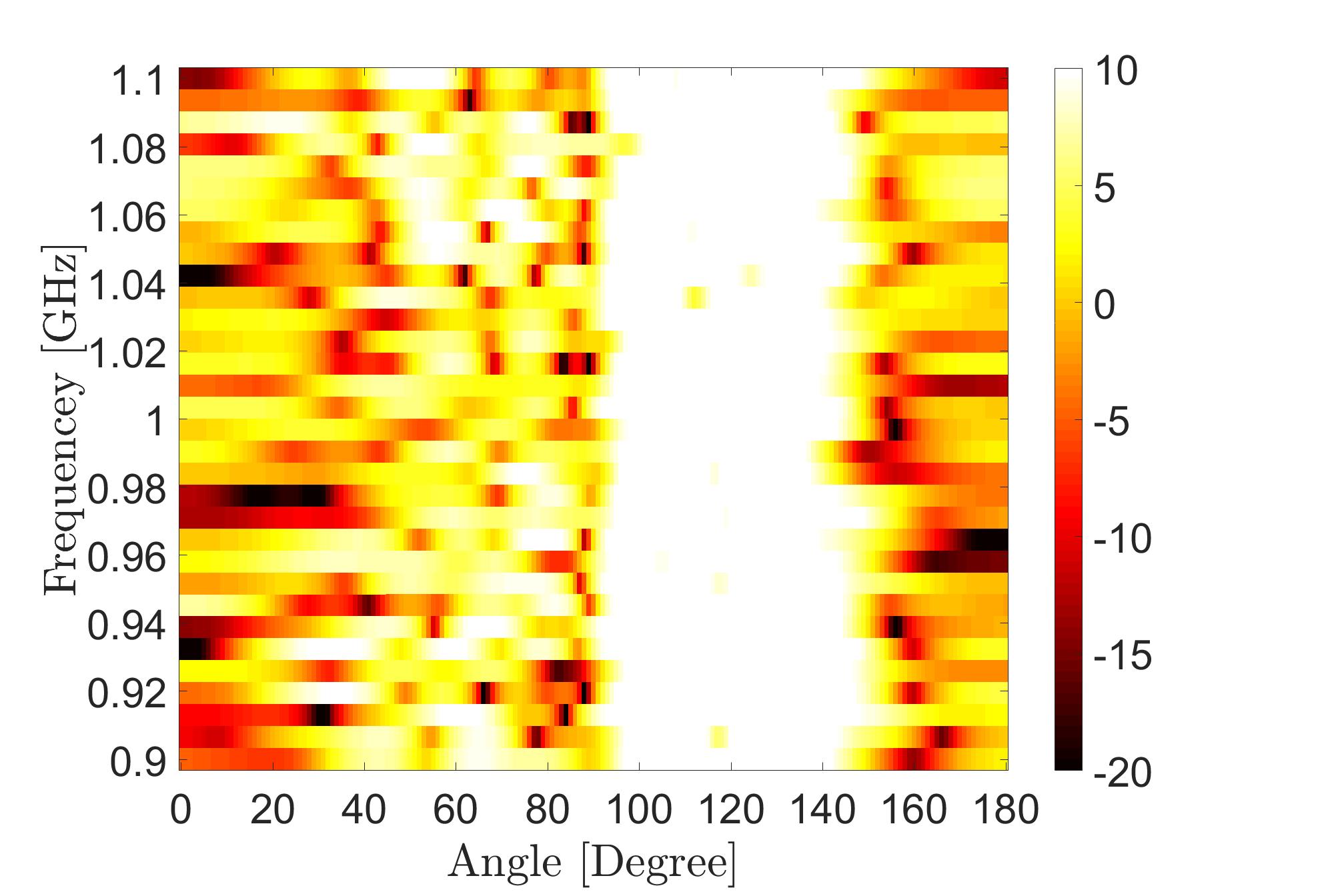}}\quad
	\subfloat[IA-CPC \cite{cui2017quadratic}]{\includegraphics[width=0.49\linewidth]{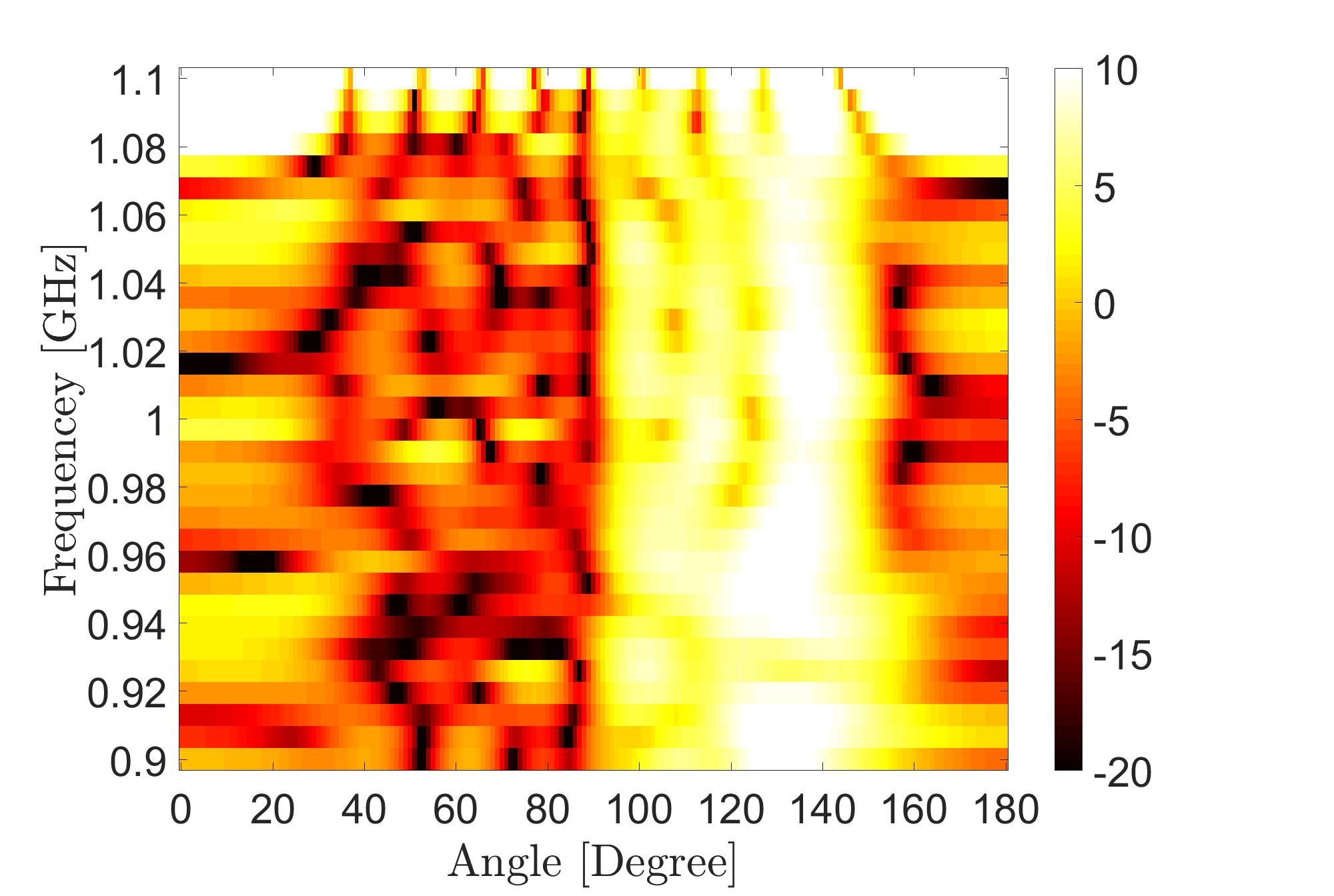}}\\
	\subfloat[ADMM\cite{Liang16}]{\includegraphics[width=0.49\linewidth]{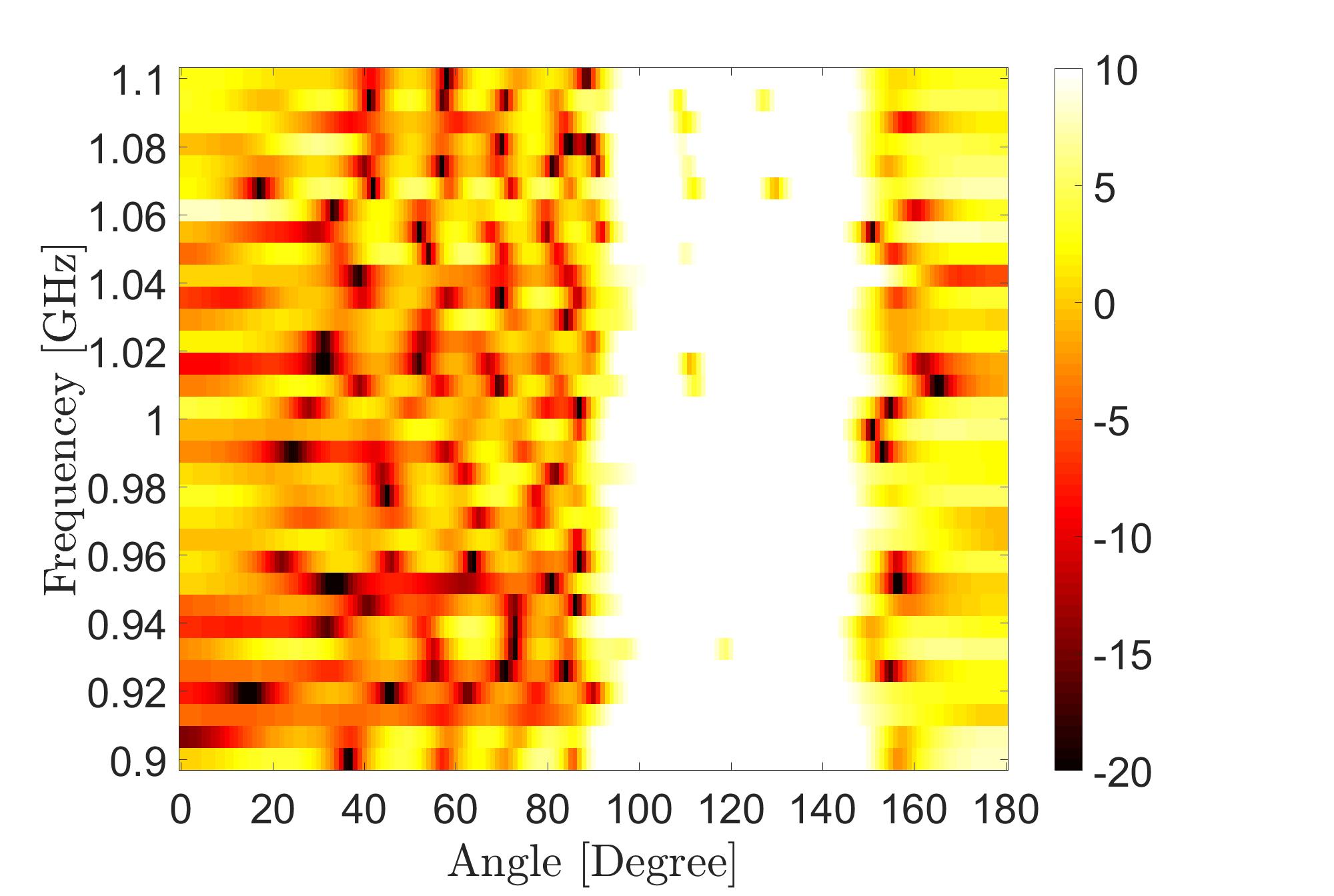}}\quad
	\subfloat[PDR]{\includegraphics[width=0.49\linewidth]{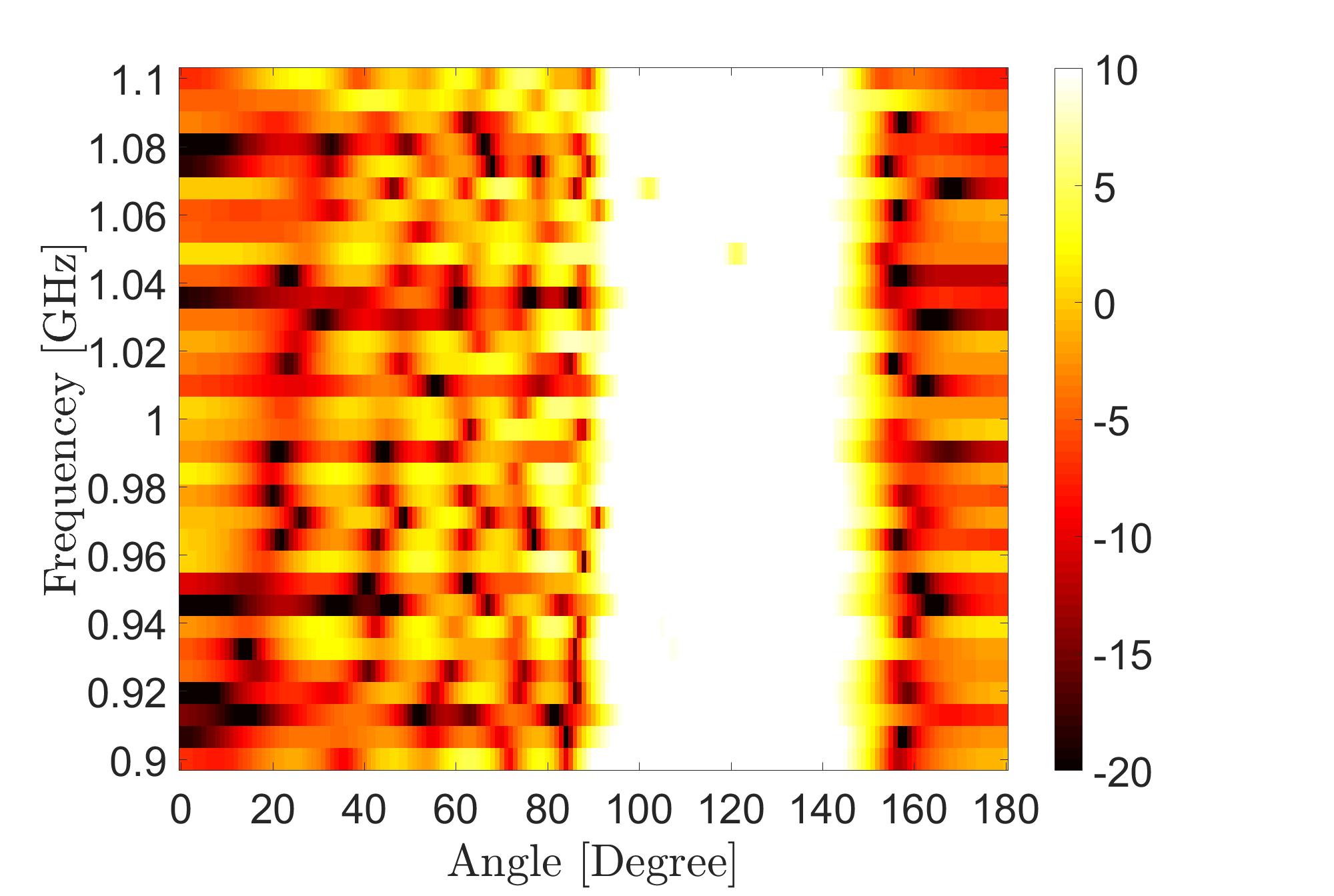}}
	\caption{The beampattern for Case 1 obtained by (a) the unconstrained design, (b) WBFIT, (c) SDR, (d) IA-CPC, (e) ADMM, and (f) PDR.} \label{Fig:BPcase1}
\end{figure*}
\section{Numerical Results} 
\label{Sec:Experiments} 

Various numerical simulations are provided to assess the performance of PDR based beampattern design and compare it to  state-of-the-art approaches. Results from the following three simulations are reported next: 1) beampattern design under the constant modulus constraint, 2) beampattern design under both constant modulus and orthogonality constraints, and 3) an investigation to examine the robustness of the produced waveforms under the two constraints to direction mismatch.

Unless otherwise specified, the following settings are used in this section. Consistent with past work \cite{he2011wideband,aldayel2017tractable}, we assume a ULA MIMO radar system with the following parameters: the number of transmit antennas $M=10$, the number of time samples $N=32$, carrier frequency $f_c=1$ GHz, bandwidth $B=200$ MHz, sampling rate $T_s=1/B$, inter-element spacing $d = c/2(f_c+B/2)$, and the spatial angle $\theta$ is divided into $S=180$ points.  
\subsection{Beampattern design under CMC}
\label{SubSec:Experiments1} 
We examine the beampattern design problem under CMC and compare the performance of PDR to the following state-of-the-art approaches: 1) Wideband Beampattern Formation via Iterative Techniques (WBFIT) \cite{he2011wideband},  2) Semi-Definite relaxation with Randomization (SDR) \cite{luo2010semidefinite,cui2014mimo}, 3)  Iterative Algorithm for Continuous Phase Case (IA-CPC) \cite{cui2017quadratic}, and 4) Design algorithm based on Alternating
Direction Method of Multipliers (ADMM) \cite{Liang16}. Similar to our work, SDR and IA-CPC are devised as approaches that optimize quadratic cost functions while enforcing CMC. That is, IA-CPC and SDR will be applied to optimize the same cost function that PDR also addresses. Finally, we also report results for the case where no constraints are posed on the waveform code $\mb x$. This {\em unconstrained design} is impractical but provides a bound on the performance of all constrained methods. For a fair comparison, PDR and the competing methods are initialized with the same waveform; a pseudo-random vector of unit magnitude complex entries.

We consider three distinct specifications of the desired beampattern. Case 1 (based on \cite{he2011wideband}) only has angular dependence and has been specified to uniformly illuminate a broadside region. Case 2 (based on \cite{aldayel2017tractable}) has both angle and frequency dependence. Case 3 has more specifications such as restricting the transmission to be in a certain frequency band for spectrally crowded scenarios \cite{rowe2014spectrally,kang2018spatio}. The step size for PDR  was chosen as $\beta = 0.00005$, $\beta  = 0.00004$, and $\beta  = 0.00004$ for Case 1, Case 2, and Case 3, respectively. The parameters for all competing methods were set as prescribed in their respective papers or by using code given directly by the authors.

	\textbf{Case 1:} The desired beampattern is given by 
	\begin{equation}\label{Scenario1}
	d(\theta,f)=
	\begin{cases}
	1 & \theta=[95^{\circ},145^{\circ}]\\
	0 & \text{Otherwise}.
	\end{cases}
	\end{equation} 
\begin{table}[H]
	\caption{The deviation from the desired beampattern (the cost function in Eq (\ref{Eq:BP})) for Case 1.} \label{Case1}
	\centering
	\begin{tabular}{|c|c|c|c|}
		\hline
		\textbf{Method}                             & $10 \ \text{log}_{10}(\rho(\mb x))$        & Run time (sec)& Iterations\\ \hline
        Unconstrained                               & $19.93$                                    & -                   & -       \\ \hline
		WBFIT\cite{he2011wideband}                  & $31.93$                                    & $0.37$              & $135$     \\ \hline
		SDR\cite{luo2010semidefinite,cui2014mimo}   & $25.50$                                     & $1107$              & $31$      \\ \hline
		IA-CPC \cite{cui2017quadratic}              & $28.30$                                    & $14.39$             & $172$     \\ \hline
		ADMM\cite{Liang16}                          & $24.93$                                    & $19.20$             & $200$     \\ \hline
		PDR                                         & $\mb{22.80}$                               & $\mb{9.27}$         & $121$     \\ \hline
	\end{tabular}
\end{table}
In Table \ref{Case1}, the values of the deviation from the desired beampattern $\rho(\mb x)$ (where $\rho (\mb x)$ is the cost function in Eq (\ref{Eq:BP})) are reported. Table \ref{Case1} confirms that PDR provides substantial gains, about 2.13 dB, 2.7 dB, and 5.5 dB over ADMM, SDR, and IA-CPC, respectively.  

In Fig. \ref{Fig:BPcase1} a 2D visualization of the designed beampattern is shown for each of the competing methods. Clearly, PDR achieves a beampattern that is closest to the unconstrained case, which naturally serves as a bound on the performance.


	\textbf{Case 2:} The desired beampattern is given by 
	\begin{equation}\label{Scenario2}
	d(\theta,f)=
	\begin{cases}
	0 & \theta=[10^{\circ},80^{\circ}], \  -\frac{B}{2}+f_c\leq f \leq f_c\\
	0 & \theta=[95^{\circ},145^{\circ}], \  f_c\leq f \leq \frac{B}{2}+f_c\\
	1 & \text{Otherwise}.
	\end{cases}
	\end{equation}

\begin{table}[H]
	\caption{The deviation from the desired beampattern (the cost function in Eq (\ref{Eq:BP})) for Case 2.} \label{Case2}
	\centering
	\begin{tabular}{|c|c|c|c|}
		\hline
		\textbf{Method}                             & $10 \ \text{log}_{10}(\rho(\mb x))$        & Run time (sec) & Iterations\\ \hline
		Unconstrained                               & $19.52$                                    & -                    & -         \\ \hline
		WBFIT\cite{he2011wideband}                  & $33.11$                                    & $0.39$               & $199$     \\ \hline
		SDR\cite{luo2010semidefinite,cui2014mimo}   & $28.41$                                    & $1190$               & $33$      \\ \hline
		IA-CPC \cite{cui2017quadratic}              & $30.86$                                    & $14.47$              & $180$     \\ \hline
		ADMM\cite{Liang16}                          & $28.91$                                    & $20.36$              & $208$     \\ \hline
		PDR                                         & $\mb{26.69}$                               & $\mb{7.38}$          & $115$     \\ \hline
	\end{tabular}
\end{table}

\begin{figure*}
	\centering
	\subfloat[Unconstrainted]{\includegraphics[width=0.49\linewidth]{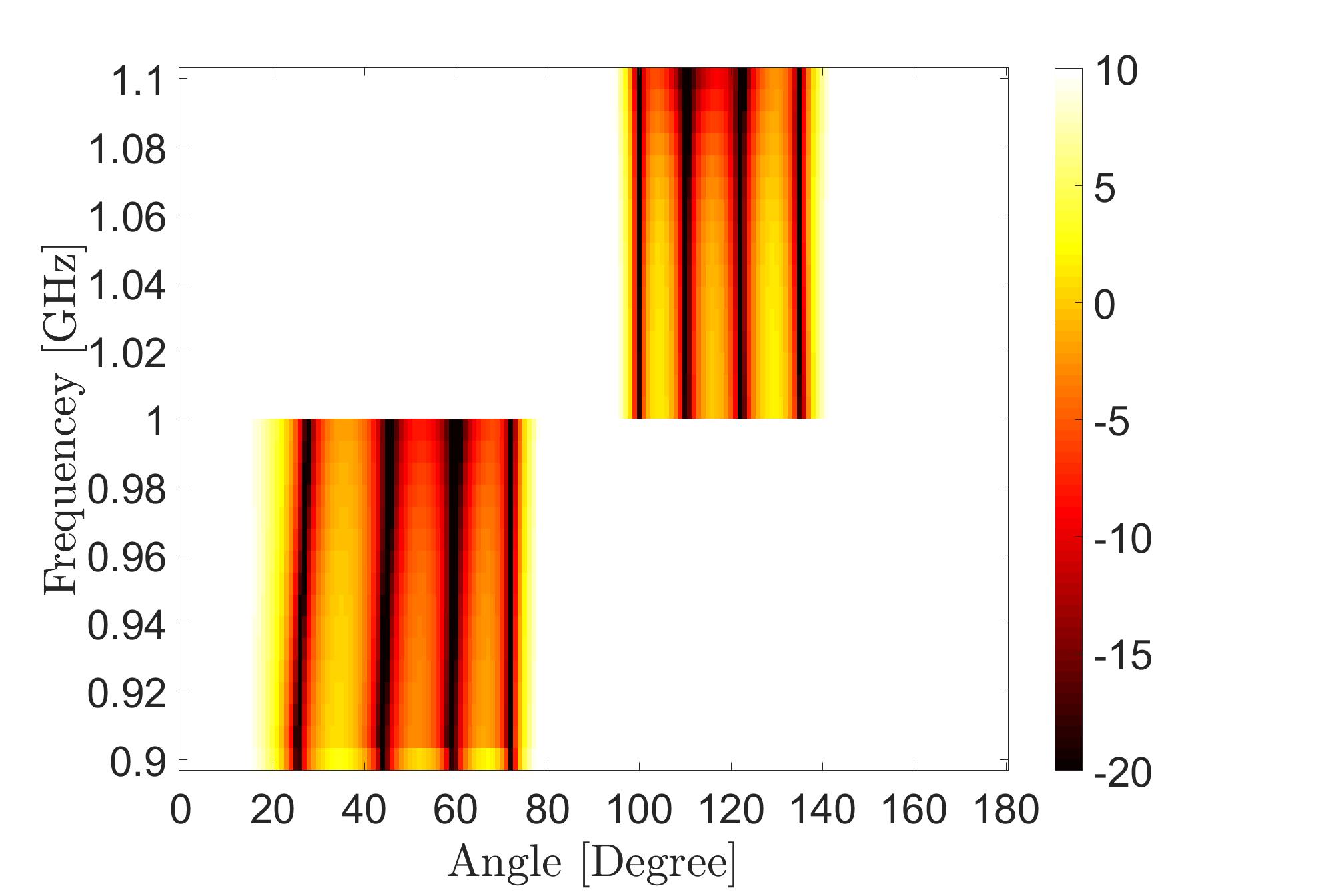}\label{desired2}}\quad
	\subfloat[WBFIT\cite{he2011wideband}]{\includegraphics[width=0.49\linewidth]{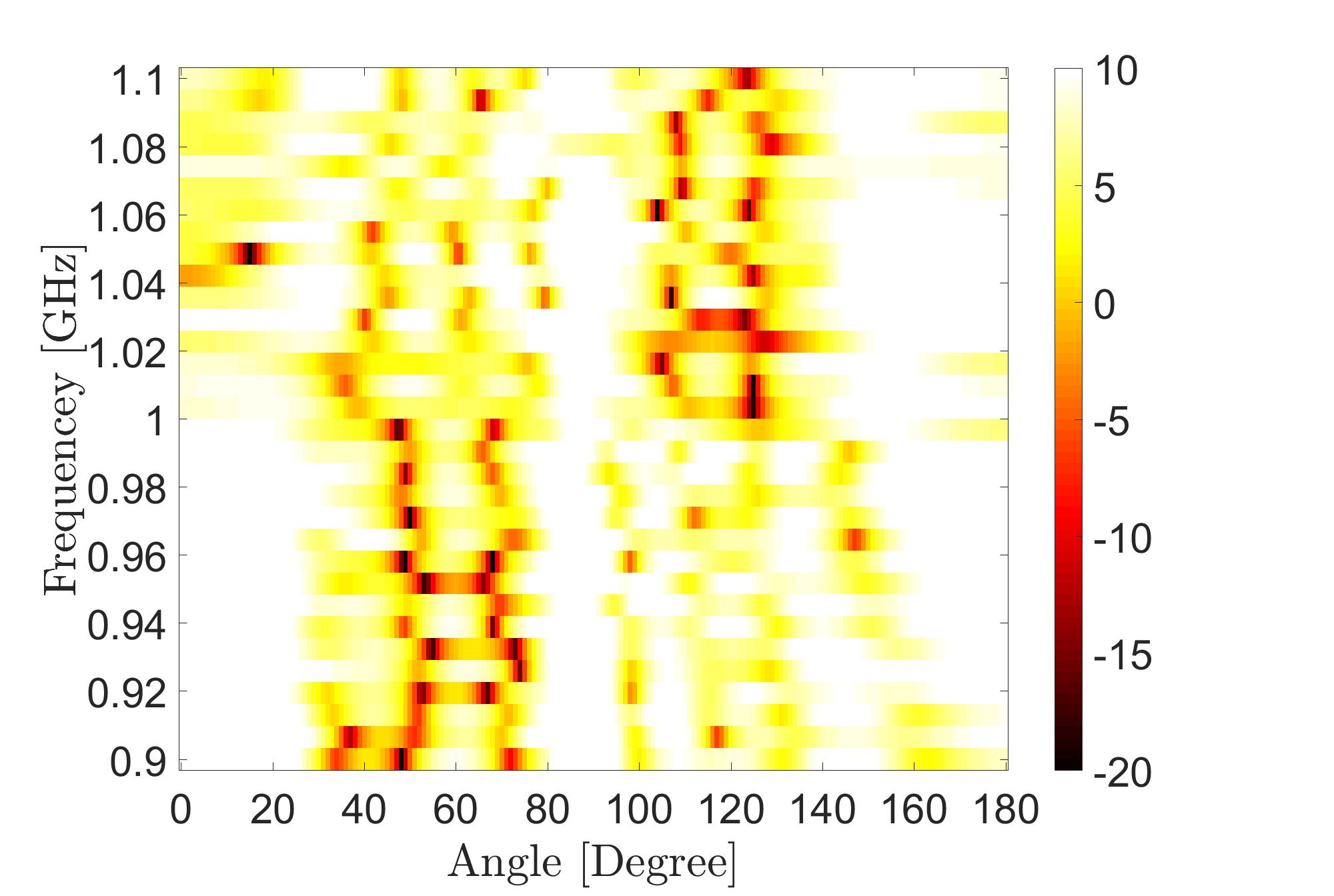}}\\
	\subfloat[SDR\cite{luo2010semidefinite,cui2014mimo}]{\includegraphics[width=0.49\linewidth]{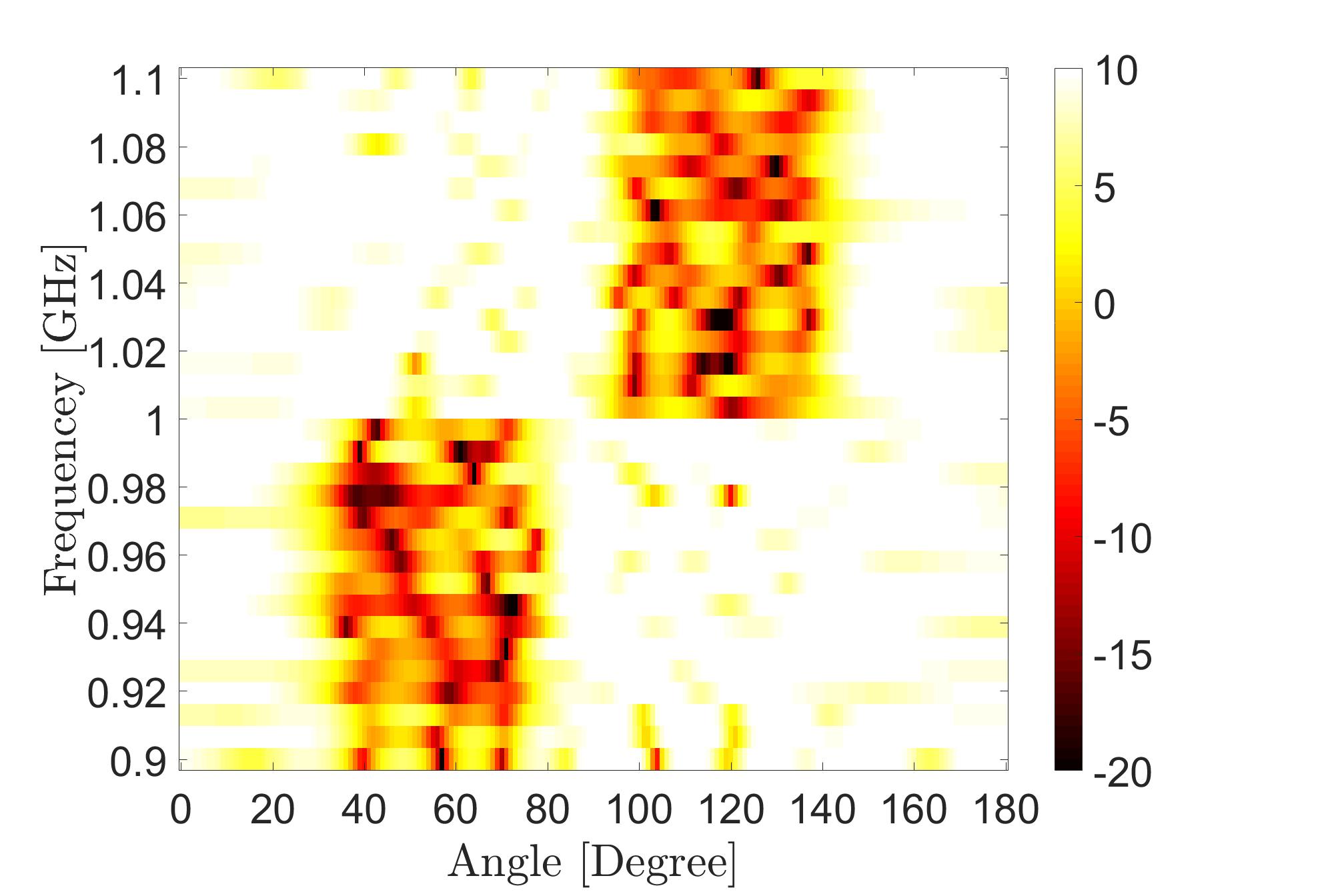}}\quad
	\subfloat[IA-CPC \cite{cui2017quadratic}]{\includegraphics[width=0.49\linewidth]{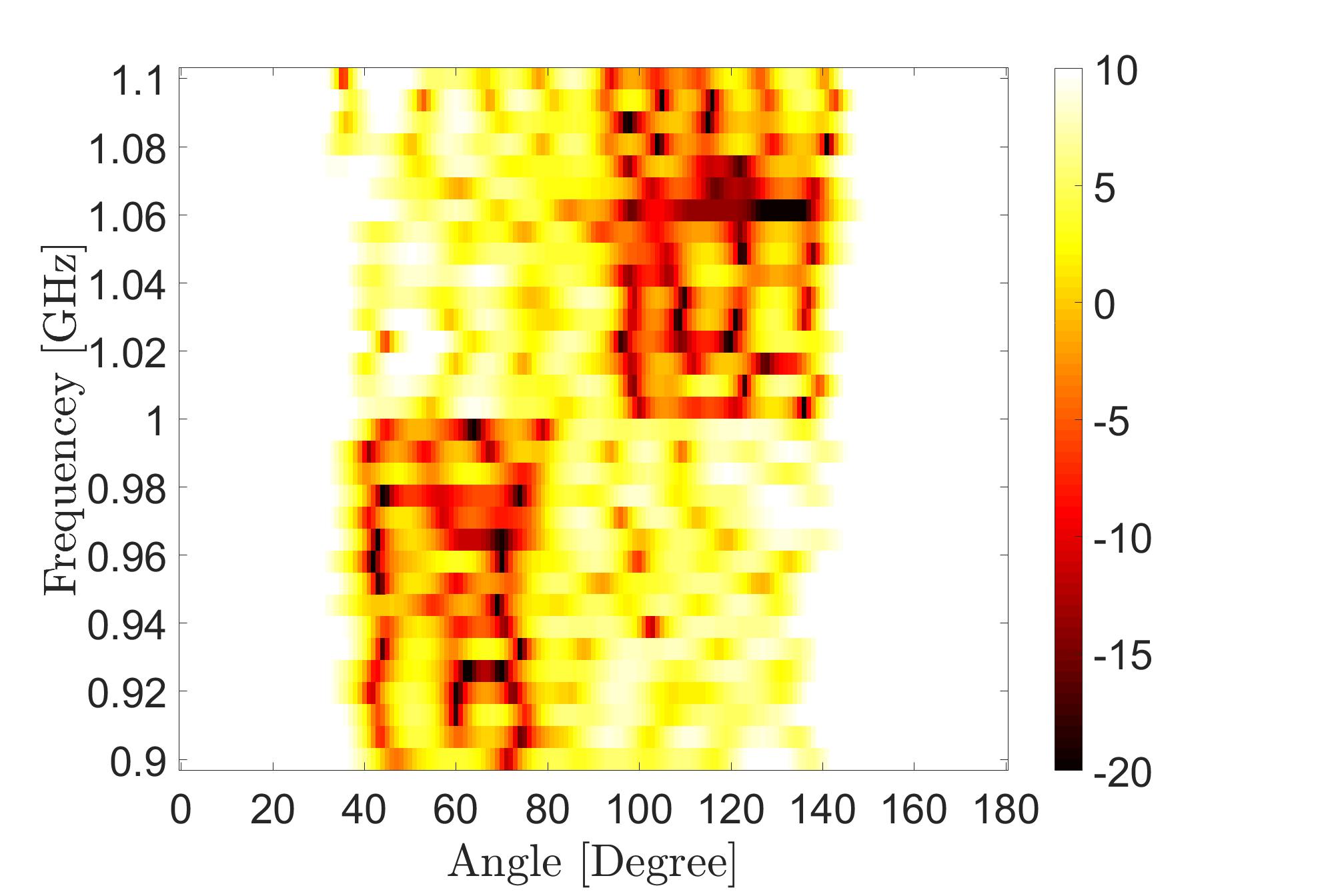}}\\
	\subfloat[{ADMM}\cite{Liang16}]{\includegraphics[width=0.49\linewidth]{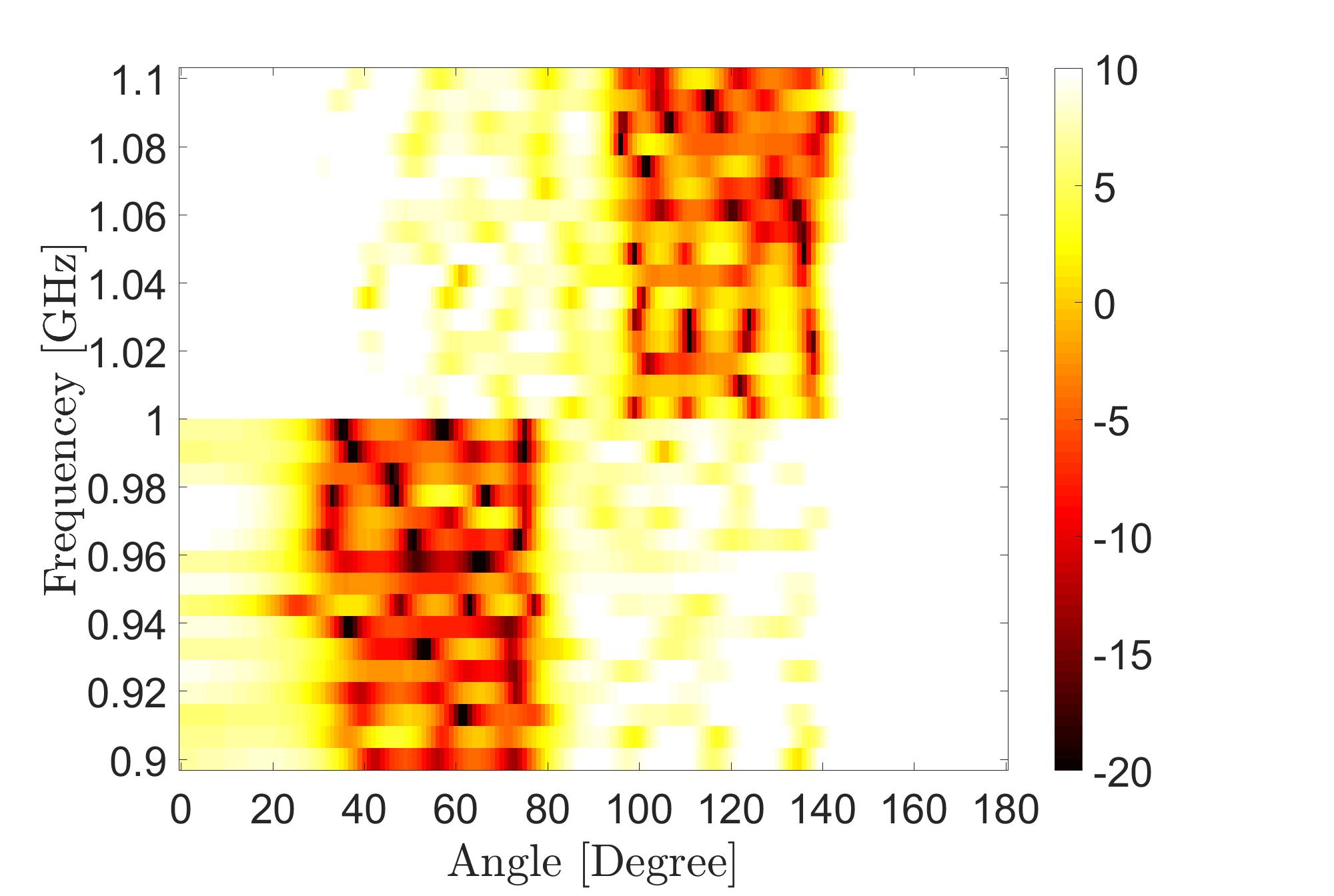}}\quad
	\subfloat[PDR]{\includegraphics[width=0.49\linewidth]{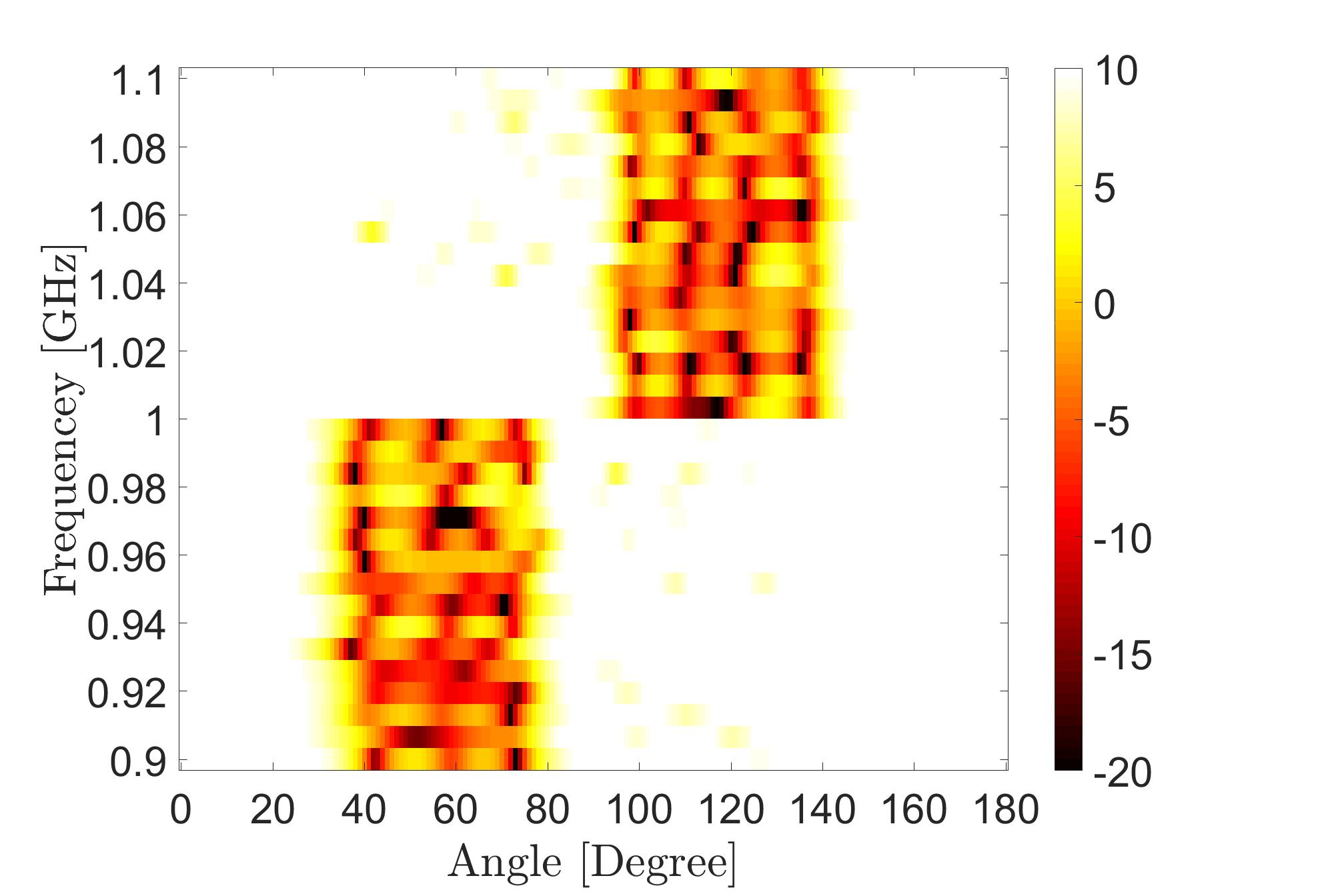}}
	\caption{{The beampattern for Case 2 obtained by (a) the unconstrained design, (b) WBFIT, (c) SDR, (d) IA-CPC, (e) ADMM, and (f) PDR.}} \label{Fig:BPcase2}
\end{figure*}
Similar to Case 1, the values of the deviation from the desired beampattern are listed in Table \ref{Case2}. Clearly, the PDR algorithm gives the closest value to the unconstrained case with a gap of 2 dB over the second best method. In Fig. \ref{Fig:BPcase2}, the designed beampattern is visualized for all the competing methods. Clearly, the beampattern that results from PDR is closer to the desired one than those resulting form competing methods. 

	\textbf{Case 3:} For this scenario, the beampattern will be suppressed in two angular-frequency regions as follows
		\begin{equation} \label{ScenarioCase3}
		d(\theta,f)=
		\begin{cases}
		0 & \theta=[40^{\circ},80^{\circ}], \  f=[943.75,981.25 ]\\
		0 & \theta=[120^{\circ},160^{\circ}], \  f=[962.5,1000]\\
		1 & \text{Otherwise}.
		\end{cases}
		\end{equation}
The transmission is restricted in certain practical frequency bands in accordance with \cite{rowe2014spectrally,kang2018spatio}. Precisely, these frequencies are: $\  f=[1.025\ \text{GHz},1.0625\ \text{GHz}]$ (see Fig. \ref{Fig:desiredCase3}).
\begin{figure}
	\centering
	\includegraphics[width=.98\linewidth]{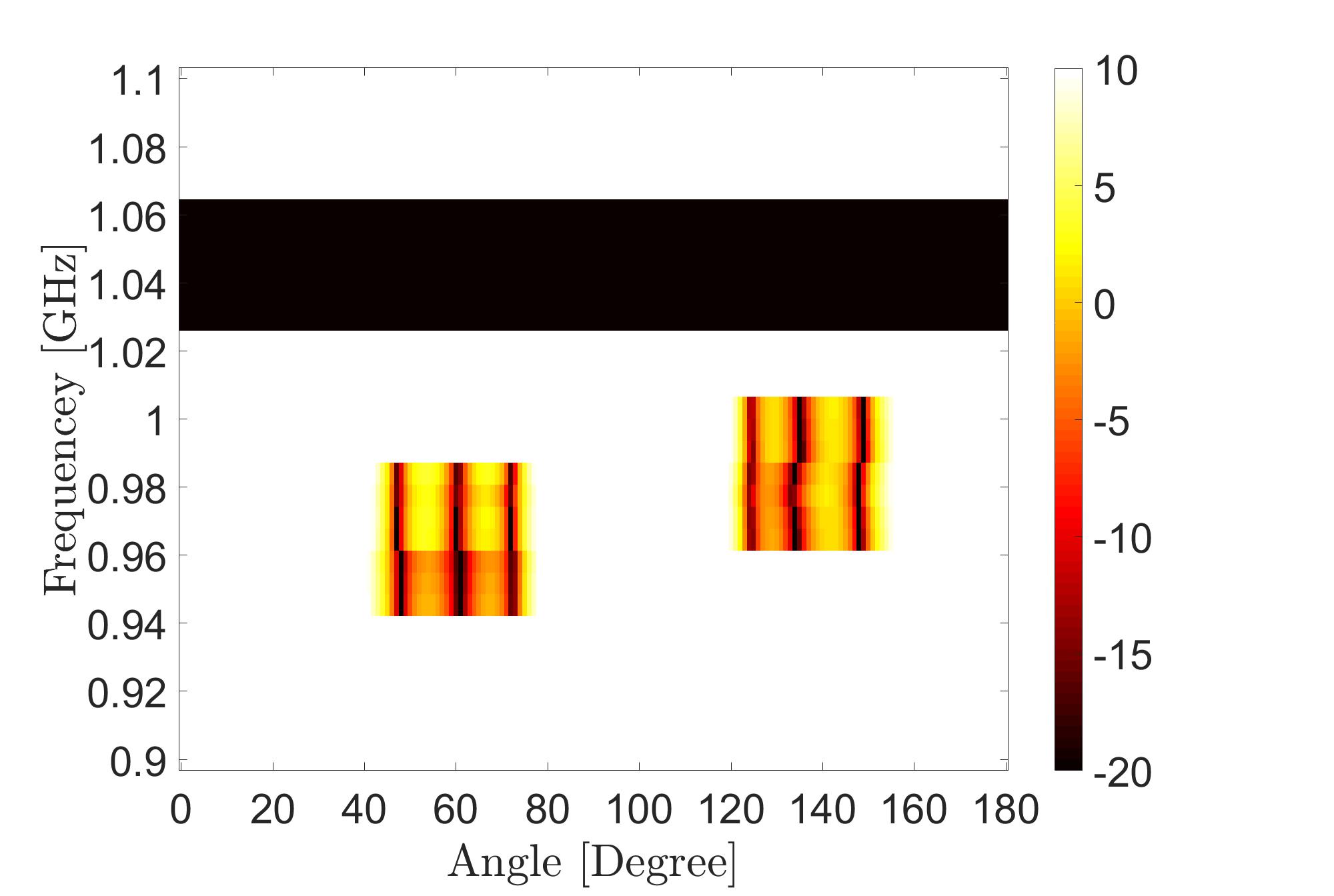}
	\caption{The desired beampattern (Case 3).}
	\label{Fig:desiredCase3}
\end{figure}	
\begin{table}[H]
	\caption{The deviation from the desired beampattern (the cost function in Eq (\ref{Eq:BP})) for Case 3.} \label{Case3}
	\centering
	\begin{tabular}{|c|c|c|c|}
		\hline
		\textbf{Method}                            & $10 \ \text{log}_{10}(\rho(\mb x))$        & Run time (sec) & Iterations\\ \hline
		Unconstrained                              & $18.18$                                    & -                    & -         \\ \hline
		WBFIT\cite{he2011wideband}                 & $34.29$                                    & $1.55$               & $102$     \\ \hline
		SDR\cite{luo2010semidefinite,cui2014mimo}  & $28.21$                                    & $1262$               & $37$      \\ \hline
		IA-CPC \cite{cui2017quadratic}             & $31.78$                                    & $18.58$              & $192$     \\ \hline
		ADMM\cite{Liang16}                         & $28.26$                                    & $24.10$              & $231$     \\ \hline
		PDR                                        & $\mb{27.86}$                               & $\mb{12.14}$         & $142$     \\ \hline
	\end{tabular}
\end{table}
As Table \ref{Case3} reveals, also in this case, PDR outperforms all competing methods in terms of deviation from the desired beampattern. 

\noindent \textbf{Remark:} Note that Tables \ref{Case1}, \ref{Case2}, and \ref{Case3} also report run time (in seconds) as the time taken to optimize the waveform code $\mathbf x$. For fairness of comparison, we used the same platform for all implementations: MATLAB R16, CPU Core i5, 3.1 GHz and 8 GB of RAM. Overall, Tables \ref{Case1}, \ref{Case2}, and \ref{Case3} show that the different methods exhibit complementary merits. WBFIT, one of the earliest methods proposed to solve this problem is the fastest, but its deviation from the idealized beampattern is highest. With more sophisticated optimization techniques, IA-CPC and SDR offer performance gains but also increase complexity.
It is readily apparent from these tables that PDR is highly efficient computationally bettered only by WBFIT. And PDR offers nearly 7 dB of gain in performance over WBFIT. Hence the results corroborate our assertion that PDR provides the most favorable complexity-performance trade-off. 

\subsection{Joint CMC and orthogonality constraints}
\label{Subsec:CMC_orthogonality}
In this numerical simulation, the beampattern design under the CMC and orthogonality constraints using PDR (run with a step size $\beta  = 0.00003$) will be examined. The level of orthogonality will be measured by using the following quantity:
\be \label{mesure}
\text{ISL}_0=20\ \text{log}_{10}\frac{\|\mathbf{X}^H\mathbf{X}-N\mathbf{I}_M\|_F}{\sqrt{MN^2}}
\ee
where ISL is the Integrated Sidelobe Level between the transmitted signal defined in \cite{he2009designing}. Orthogonality across antennas is equivalent to the auto correlation case (ISL at time-lag $0$), i.e., $\text{ISL}_0$.
In terms of the desired beampattern, we consider the same scenario as Case 1 in the previous subsection. The values of the deviation from the desired beampattern are reported in Table \ref{Table:CMC_orthogonality}. PDR is now compared against approaches that also directly or approximately capture both CMC and orthogonality; this includes: 1) the well-known and widely used linear frequency modulated (LFM ) waveform code  \cite{levanon2004radar,richards2010principles}, 2) Weighted-cyclic algorithm-new (WeCAN) \cite{he2009designing}, and 3) the recent simulated annealing based approach (SimulAnneal) \cite{deng2016mimo}, which is one of the few known techniques that performs explicit beampattern design under both CMC and orthogonality constraints.

First, compared to Table \ref{Case1}, the deviation as reported in Table \ref{Table:CMC_orthogonality} is higher. This is of course to be expected because we are not just enforcing CMC, but also orthogonality which means a smaller feasible set of waveform codes to optimize over. The results in Table \ref{Table:CMC_orthogonality} reveal that LFM and WeCAN lead to somewhat high deviation - this is unsurprising since neither approach explicitly designs the beampattern. When $\alpha = 100$, PDR leads to waveform codes that are orthogonal for all practical purposes as evidenced by the ISL measure. As  Table \ref{Table:CMC_orthogonality} reveals, PDR achieves the closest beampattern to the desired one (lowest deviation in dB) with SimulAnneal as the second best. Remarkably, even with $\alpha = 200$, i.e. when the emphasis on orthogonality is particularly strong (ISL of $-12.12$ dB for PDR vs. $-3.67$ for SimulAnneal), PDR provides  a gain of 1 dB. The gain of PDR over SimulAnneal  is about 3 dB for comparable ISL values.
\begin{table}[!ht]
	\centering
	\caption{Cost function and Auto-correlation level for PDR vs WeCAN and SA-Method}
	\label{Table:CMC_orthogonality}
	\begin{tabular}{|c|c|c|c|c|}
		\hline
		Method                             & $10 \ \text{log}_{10}(\rho(\mb x))$ & $\alpha$ & $\text{ISL}_0$ (dB) & Time (sec) \\ \hline
		LFM                                & $32.93$                             & -        & $-302.53$           & -                \\ \hline
		WeCAN \cite{he2009designing}       & $32.71$                             & -        & $-102.92$           & $11.52$          \\ \hline
		SimulAnneal \cite{deng2016mimo}    & $30.72$                             & -        & $-3.67$             & $9.29$           \\ \hline
		\multirow{3}{*}{\textbf{PDR}}
		                                   & $27.77$                             & $80$     & $-5.27$             & $10.58$          \\ \cline{2-5} 
		                                   & $28.64$                             & $100$    & $-6.92$             & $12.67$          \\ \cline{2-5} 
		                                   & $29.66$                             & $200$    & $-12.12$            & $15.05$          \\ \hline
	\end{tabular}
\end{table}
\subsection{Benefits of orthogonality: robustness to direction mismatch}
The transmit-receive pattern $G_{TR}(\theta,\theta_0)$  measures the beamformer response when the beam is digitally steered in the direction $\theta$ and when the true location of the target
is at angle $\theta_0$. It has been shown in \cite{bekkerman2006target} and \cite{li2009mimo} (as an advantage of orthogonal over coherent signals) that for orthogonal signalling, the effect of the direction mismatch (when the target is not located in the center of
the transmit beam) is minimal. 

The pattern $G_{TR}(\theta,\theta_0)$ can be expressed as  
\be \label{bp1}
\begin{aligned}
	G_{TR}(\theta,\theta_0)=N\frac{|\mathbf{a}_R^H(\theta)\mathbf{a}_R(\theta_0)|^2}{M_R}.\frac{|\mathbf{a}_T^H(\theta)\mathbf{R}_s^T\mathbf{a}_T(\theta_0)|^2}{\mathbf{a}_T^H(\theta)\mathbf{R}_s^T\mathbf{a}_T(\theta)}
\end{aligned}
\ee 
where $M_R$ is the number of receiver antennas, $\mathbf{a}_T(\theta)$ and $\mathbf{a}_R(\theta)$ are the steering vectors on the transmitter and receiver sides, respectively. These steering vectors are defined in Equations (4.6) and (4.7) in \cite{li2009mimo}. 
$\mathbf{R}_s$ is the transmit signal correlation matrix defined as 
\be \label{coherencematrix}
\mathbf{R}_s = \big(\mathbf{X}^T \odot \mathbf{a}_T(\theta_0)\big)\big(\mathbf{X}^T \odot \mathbf{a}_T(\theta_0)\big)^H
\ee
where $\mb X \in \mathbb{C}^{N\times M}$ is the transmit waveform matrix defined in  Section \ref{Sec:PDR_Orth}. For the two extreme cases 1) coherent transmission\footnote{In coherent transmission, the transmit signals from all antennas are phase-shifted versions from one reference signal \cite{li2009mimo}.} ($\mathbf{R}_s ={\mathbbm{1}_M}$) and  2) orthogonal transmission ($\mathbf{R}_s =\mb I_M$ )  \cite{bekkerman2006target}, the transmit-receive patterns for these cases will be
\ben \label{bp2}
\begin{aligned}
	G_{TR(\text{coherent})}(\theta,\theta_0)=N\frac{|\mathbf{a}_R^H(\theta)\mathbf{a}_R(\theta_0)|^2|\mathbf{a}_T^H(\theta_0) \mb 1|^2}{M_R}
\end{aligned}
\een 
\ben \label{bp3}
\begin{aligned}
	G_{TR(\text{orthogonal})}(\theta,\theta_0)=N\frac{|\mathbf{a}_R^H(\theta)\mathbf{a}_R(\theta_0)|^2|\mathbf{a}_T^H(\theta)\mathbf{a}_T(\theta_0)|^2}{M_RM}
\end{aligned}
\een 
Fig. \ref{Fig:mismatch} shows the pattern $G_{TR}(\theta,\theta_0)$ for (a) LFM ($\mathbf{R}_s=N\mathbf{I}_M$), (b) Coherent Transmission, (c) WBFIT, (d) IA-CPC, and (e) PDR ($\alpha = 200$) signals where the beam in the transmit mode is directed to $\theta=125^{\circ}$ and the true location of the target is such that $\Delta \theta = (\theta - \theta_0)=0^{\circ},\ 10^{\circ},\ 20^{\circ}$. The desired beampattern for this {simulation} is the same as Case 1 in Section IV-A, except that PDR was now optimized to generate orthogonal waveforms.

The advantage of using orthogonal signals (LFM) over a coherent scheme is very noticeable: the gain loss is nonexistent for LFM when the target deviates from the transmission direction (i.e., target is not located in the center of the transmit beam).  Remarkably, PDR signals with $\alpha=200$ achieve results comparable to the LFM case, whereas in the absence of orthogonal processing, WBFIT in Fig. \ref{Fig:mismatch} (c) and IA-CPC in Fig. \ref{Fig:mismatch} (d) like the coherent case in Fig. \ref{Fig:mismatch} (b) suffer significant loss in mainlobe strength under target direction mismatch. 

{The focus of our work is on transmit waveform design but in future investigations, receive processing may also be optimized to obtain the most desirable transmit-receive beampattern.}
\begin{figure}[!ht]
	\centering
	\subfloat[LFM]{\includegraphics[width=0.9\linewidth]{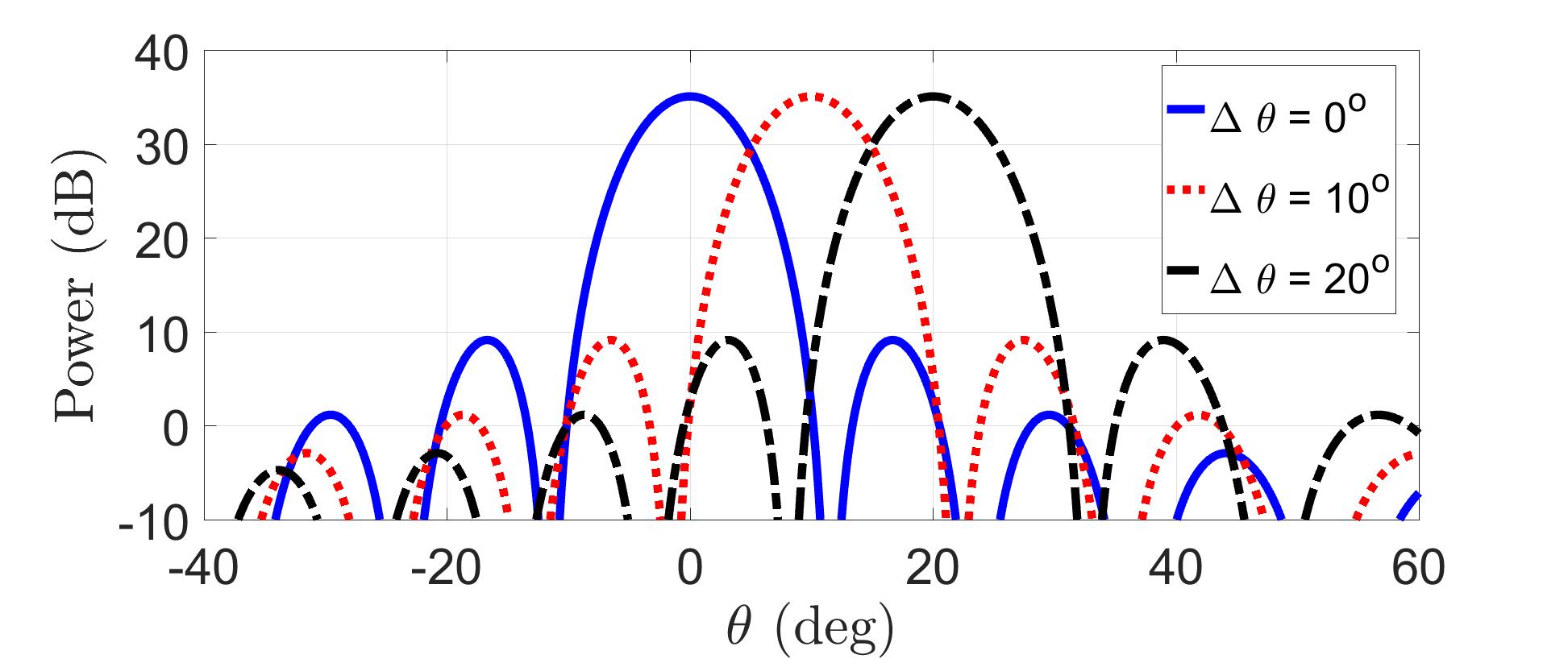}}\\
	\subfloat[Coherent Transmission]{\includegraphics[width=0.9\linewidth]{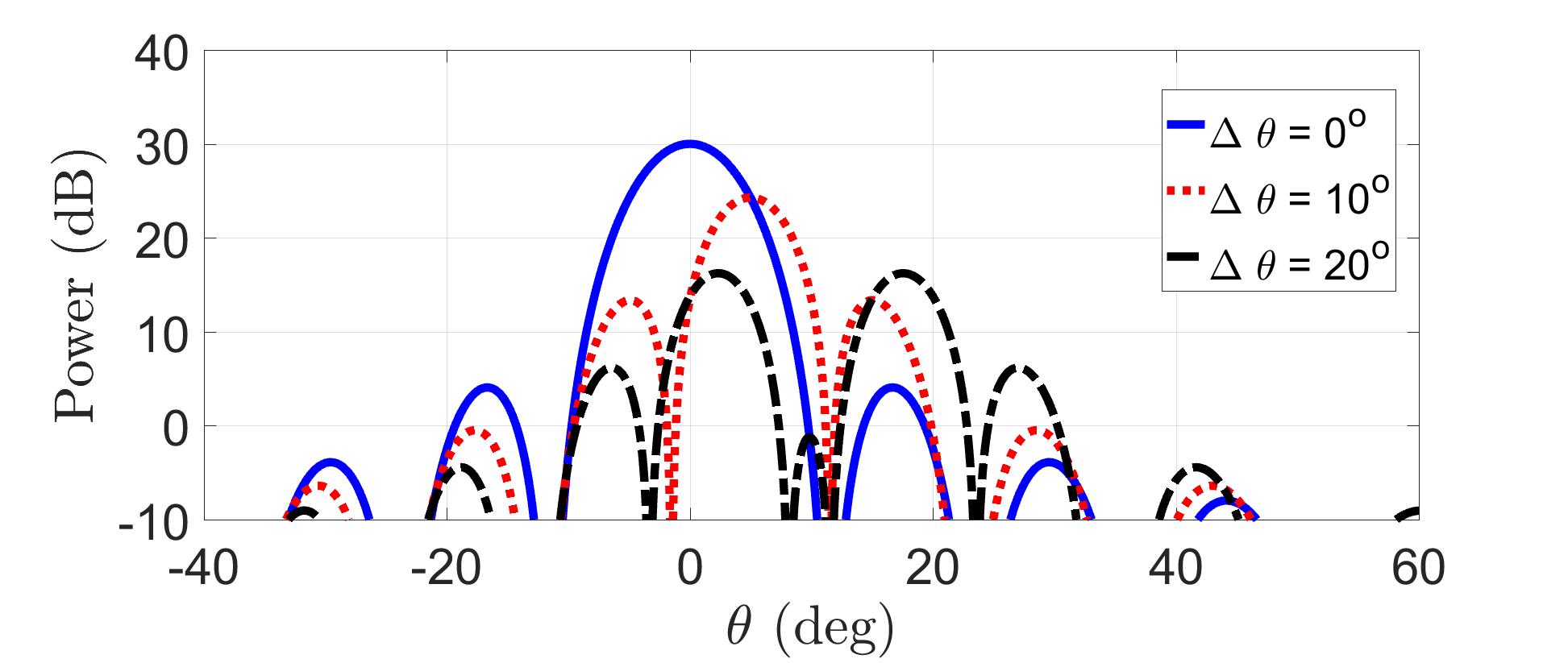}}\\
	\subfloat[WBFIT \cite{he2011wideband}]{\includegraphics[width=0.9\linewidth]{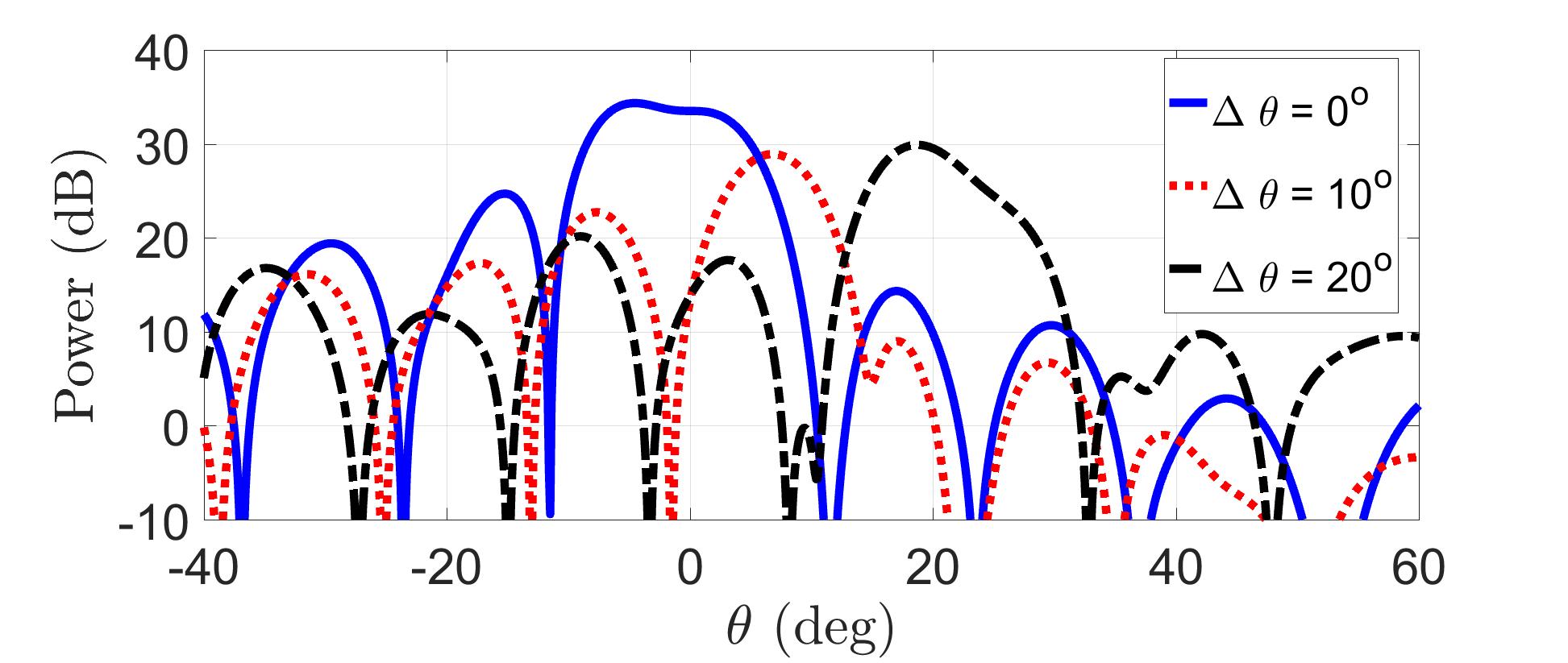}}\\
	\subfloat[IA-CPC \cite{cui2017quadratic}]{\includegraphics[width=0.9\linewidth]{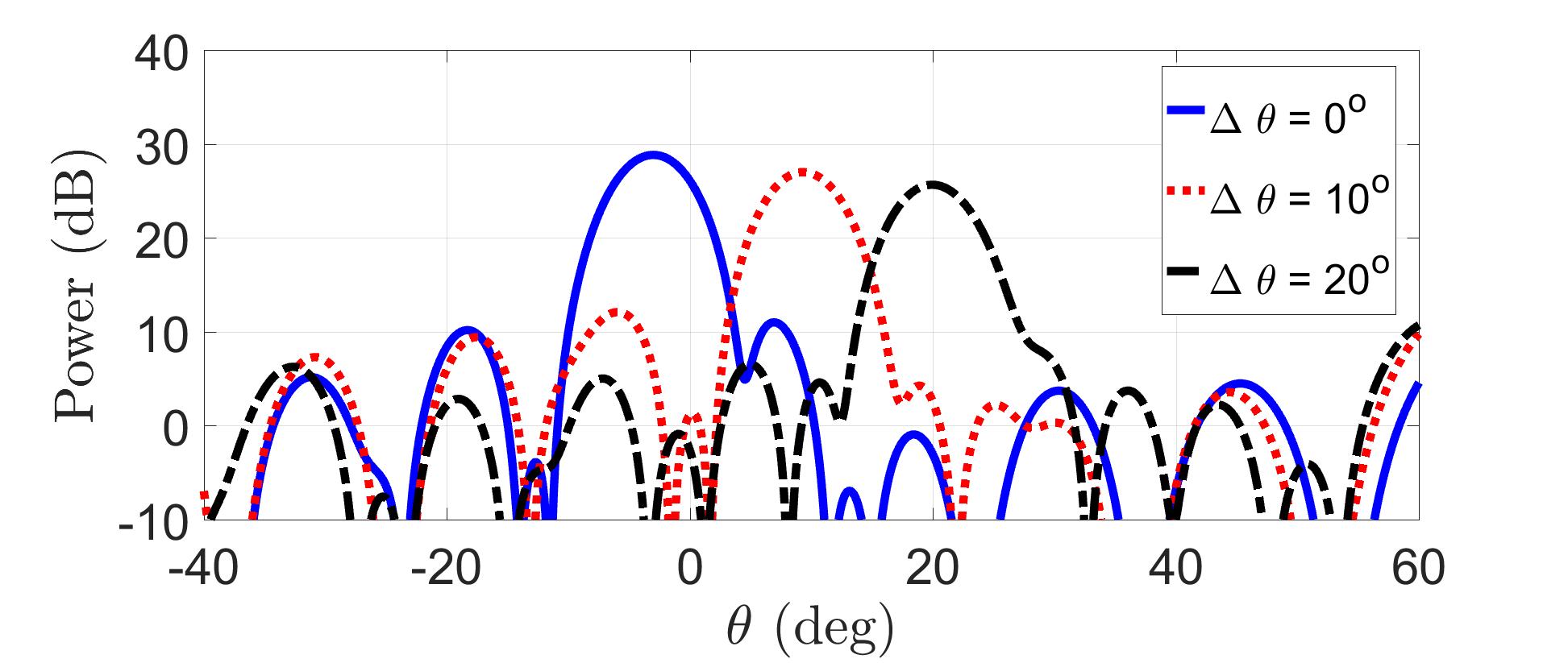}}\\
	\subfloat[PDR $\alpha = 200$]{\includegraphics[width=0.9\linewidth]{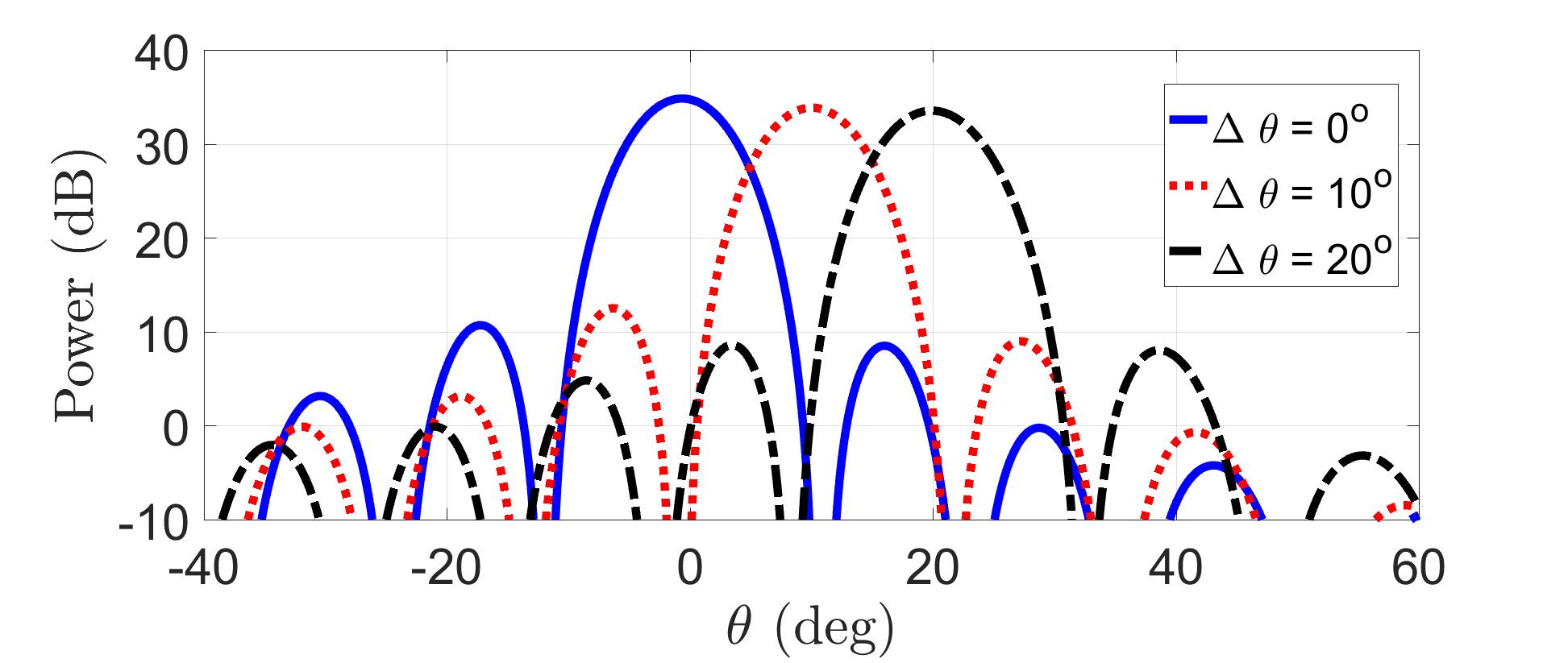}}\label{Fig:mismatch_d}
	\caption{{The pattern $G_{TR}(\theta,\theta_d)$ for (a) LFM, (b) Coherent, (c) WBFIT, and (d) PDR ($\alpha = 200$) signals. The beam in the transmit mode is directed to $\theta=0^{\circ}$ and the target is located at $\theta_0=0^{\circ},\ 10^{\circ},\ 20^{\circ}$.}} \label{Fig:mismatch}
\end{figure}
\section{Discussions and Conclusion}
\label{Sec:Conclusion}

We consider the problem of transmit beampattern design for MIMO radar under compelling practical constraints. The non-convex constant modulus constraint (CMC) is our main focus whose presence is known to yield a hard optimization problem. For tractability, CMC is addressed in the literature often by relaxations and approximations or by approaches that are  computationally burdensome. Our proposed PDR algorithm invokes the principles of optimization over manifolds to address the non-convex CMC and we demonstrate via simulations that the said PDR offers a favorable performance-complexity trade-off. Analytical guarantees of monotonic cost function decrease and convergence are provided for quadratic cost functions that arise in beampattern design. Finally, a tractable extension was developed to incorporate orthogonality of waveforms across antennas. Experimentally, the benefit of orthogonality combined with CMC is the synthesis of practical, real-world waveforms that demonstrate robustness to target direction mismatch.

A viable future work direction is to exploit for CMC constrained beampattern design new algorithms that use the framework of the sequential approximation such as \cite{razaviyayn2013unified} and recent advances in \cite{aubry2018new}, \cite{song2015optimization}. Of particular interest is investigating KKT optimality of the resulting solution.
 
\section*{Acknowledgement}
The authors would like to thank the authors of the work \cite{cui2017quadratic} for providing us with the MATLAB code for the  Iterative Algorithm for Continuous Phase Case (IA-CPC) algorithm.  
\appendix
\section{Appendix A}
\label{Sec:Appendix}
In this section we provide the proof of the three lemmas that presented in this work and the equivalence between Eq (\ref{diffPsi}) and the retraction step.
\subsection{Proof of Lemma \ref{Lemma:L1}}
\label{Subsection:AppA}
Before we establish the proof, the projection operator in Eq (\ref{compcir2}) will be reformulated as follows. The projection of a vector $\mb w \in \mathbb{C}^{L}$  onto the tangent space $\mathcal{T}_{\mb z}\mathcal{S}^L$ at a point $\mb z \in \mathcal{S}^L$ can be rewritten using the elementary properties the Hadamard product and the fact that $\text{Re}\{\mb w^{*}\odot\mb z\}=\frac{1}{2}\big[\mb w^{*}\odot\mb z+\mb w\odot\mb z^{*}\big]$ as  

\begin{IEEEeqnarray}{rCl}
	\mathbf{P}_{\mathcal{T}_{\mb z}\mathcal{S}^L}(\mb w)&=&\mb w - \frac{1}{2}\big[\mb w^{*}\odot\mb z+\mb w\odot\mb z^{*}\big]\odot \mb z\nonumber\\
&=&\mb w - \frac{1}{2}\big[\mb w^{*}\odot\mb z\odot \mb z+\mb w\odot\mb z^{*}\odot \mb z\big]\nonumber\\
&=&\mb w - \frac{1}{2}\big[\text{ddiag}(\mb z \mb z^T)\mb w^{*}+\mb w\big]\nonumber\\
&=&\frac{1}{2}\big(\mb w - \mb D_{\mb z}\mb w^{*}\big)=\frac{1}{2}\big(\mb w - \tilde{\mb w}\big)\label{proj1}
\end{IEEEeqnarray}
where $\mb D_{\mb z}=\text{ddiag}(\mb z \mb z^T)$,  and $\tilde{\mb w}=\mb D_{\mb z} \mb w^{*}$. Note that since $\mb z \in \mathcal{S}^L$, then $\mb D_{\mb z}^H\mb D_{\mb z}=\mb D_{\mb z}\mb D_{\mb z}^H= \mb I$, and hence \be \label{observation1}\tilde{\mb w}^H\tilde{\mb w}=\mb w^H\mb w\ee
For notation simplicity, will use $\mb w$ instead of $\nabla_{\mb x} \bar{f}(\mb x_{(k)})$ to represent the gradient of $\bar{f}(\mb x)$ defined in Eq (\ref{Eq:P2}) at iteration $k$, i.e., $\mb w= \nabla_{\mb x} \bar{f}(\mb x_{(k)})=2\mb R \mb x_{(k)} - 2 {\mb q}$, and hence the search direction $\boldsymbol{\eta}_{(k)}$ will be 
\ben
\boldsymbol{\eta}_{(k)}= -\mb w
\een
The projection of the search direction $\boldsymbol{\eta}_{(k)}$ onto the tangent space $\mathcal{T}_{\mb x_{(k)}}\mathcal{S}^L$ at a point $\mb x_{(k)}$ will be computed using the new projection formula Eq (\ref{proj1}), and with a step $\beta$ along this direction  starting from $\mb x_{(k)}$, the update on the tangent space will be
\begin{IEEEeqnarray}{rCl}
		\bar{\mathbf{x}}_{(k)}&=&\mathbf{x}_{(k)}+\beta \mathbf{P}_{\mathcal{T}_{\mb x_{(k)}}\mathcal{S}^L}(\boldsymbol{\eta}_{(k)})\nonumber\\
	&=&\mathbf{x}_{(k)}+\beta \mathbf{P}_{\mathcal{T}_{\mb x_{(k)}}\mathcal{S}^L}(-\mb w)\nonumber\\
	&=&\mathbf{x}_{(k)} -\frac{\beta}{2}\big(\mb w - \tilde{\mb w}\big)\label{alg22}
\end{IEEEeqnarray}
\begin{proof}
	Let $\mb R = \mb P +\gamma \mb I$, the value of the cost function $\bar{f}(\mb x)=\mb x^H \mb R\mb x - \mb q^H \mb x - \mb x^H \mb q$ on the tangent space at the point $\bar{\mathbf{x}}_{(k)}$ \big(using the value of $\bar{\mathbf{x}}_{(k)}$ in Eq (\ref{alg22})\big) will be 
		\begin{IEEEeqnarray}{rCl}
			\bar{f}(\bar{\mathbf{x}}_{(k)}) &=& \mb x_{(k)}^H\mb R \mb x_{(k)}\nonumber\\* \quad
		&&-\frac{\beta}{2}\mb x_{(k)}^H \mb R (\mb w - \tilde{\mb w})-\frac{\beta}{2} (\mb w - \tilde{\mb w})^H \mb R \mb x_{(k)} \nonumber\\ 
		&&+\frac{\beta^2}{4} (\mb w - \tilde{\mb w})^H  \mb R (\mb w - \tilde{\mb w})-\mb q^H \mb x_{(k)}\nonumber \\ 
		&& + \frac{\beta}{2} \mb q^H(\mb w - \tilde{\mb w})-\mb x_{(k)}^H \mb q+\frac{\beta}{2} (\mb w - \tilde{\mb w})^H \mb q\nonumber \\
		&=& \mb x_{(k)}^H  \mb R \mb x_{(k)}-\frac{\beta}{4} \mb w^H(\mb w - \tilde{\mb w})-\frac{\beta}{4} (\mb w - \tilde{\mb w})^H \mb w \nonumber\\ 
		&&+\frac{\beta^2}{4} (\mb w - \tilde{\mb w})^H  \mb R (\mb w - \tilde{\mb w})-\mb q^H \mb x_{(k)}-\mb x_{(k)}^H \mb q \nonumber
		\end{IEEEeqnarray}
	Now, the difference  $\bar{f}(\mb x_{(k)})-\bar{f}(\bar{\mb x}_{(k)})$ will be
	\begin{IEEEeqnarray}{rCl}
		\bar{f}(\mb x_{(k)})-\bar{f}(\bar{\mb x}_{(k)})
		&=& \frac{\beta}{4} \mb w^H(\mb w - \tilde{\mb w})+\frac{\beta}{4} (\mb w - \tilde{\mb w})^H \mb w \nonumber\\
		 &&-\frac{\beta^2}{4} (\mb w - \tilde{\mb w})^H \mb R (\mb w - \tilde{\mb w})\nonumber\\
		&=&\frac{\beta}{4} \big( \mb w^H\mb w - \mb w^H\tilde{\mb w} - \tilde{\mb w}^H \mb w+ \tilde{\mb w}^H\tilde{\mb w}\big)\nonumber\\
		 &&-\frac{\beta^2}{4} (\mb w - \tilde{\mb w})^H \mb R (\mb w - \tilde{\mb w})\label{Eq:diff1_1}\\
		&=& \frac{\beta}{4}(\mb w - \tilde{\mb w})^H (\mb I-\beta \mb R)(\mb w - \tilde{\mb w})\label{Eq:diff1}
\end{IEEEeqnarray}
	where in Eq (\ref{Eq:diff1_1}) the observation in Eq (\ref{observation1}) is used. 
	
Recall the definition of the matrix $\mb P=\sum_{p}  \mb F^H_p \mb A^H_p\mb A_p \mb F_p$, for each $p$ the matrix $(\mb A_p \mb F_p)^H(\mb A_p \mb F_p)$ is positive semi definite, i.e., $\mb y^H(\mb A_p \mb F_p)^H(\mb A_p \mb F_p)\mb y=\|\mb A_p \mb F_p\mb y\|^2_2\geq 0 \ \forall \ \mb y \in \mathbb{C}^{L} $ and any non-negative linear combination of positive semidefinite matrices is positive semidefinite (Observation 7.1.3 \cite{horn1990matrix}), then $\mb P$ is positive semidefinite. Since $\mb P$ is positive semidefinite, then $\mb R$ is positive definite. Now, using Theorem 7.1 in  \cite{zhang2011matrix}, the matrix $\mb R$ can be diagonalized as 
	\be \label{decomp} \mb R=\mb U \Lambda \mb U^H \ee where $\Lambda$ is diagonal matrix with the eigenvalues of $\mb R$ in the main diagonal, and $\mb U$ is unitary, i.e., $\mb U^H \mb U=\mb U \mb U^H= \mb I$. 
	
	Using Eq (\ref{decomp}) in Eq (\ref{Eq:diff1}) yields
		\begin{IEEEeqnarray}{rCl} 
	\bar{f}(\mb x_{(k)})-\bar{f}(\bar{\mb x}_{(k)})
		&=& \frac{\beta}{4}(\mb w - \tilde{\mb w})^H(\mb I -\beta \mb U \Lambda \mb U^H)(\mb w - \tilde{\mb w})\nonumber\\
		&=& \frac{\beta}{4}(\mb w - \tilde{\mb w})^H\mb U(\mb I -\beta \Lambda) \mb U^H(\mb w - \tilde{\mb w})\nonumber\\
		&=&\frac{\beta}{4}\mb h^H(\mb I -\beta \Lambda) \mb h\label{Eq:diff2}
	\end{IEEEeqnarray}
	where $\mb h = \mb U^H(\mb w - \tilde{\mb w}) \in \mathbb{C}^L$. If the step $\beta$ is chosen in a way such that $\frac{\beta}{4}\mb h^H(\mb I -\beta \Lambda) \mb h \geq 0 \ \forall \ \mb h$ or the  matrix $(\mb I -\beta \Lambda)$ is positive semi-definite and  $\beta \geq 0$, then $\bar{f}(\mb x)-\bar{f}(\bar{\mb x})$ will be non-negative. Then, the following condition on $\beta$ must be satisfied
	\be \label{condi2}\mb I-\beta \Lambda \geq 0 \Leftrightarrow \beta \leq \frac{1}{\lambda_{\mb R}}\ee
	where $\lambda_{\mb R}$ is the largest eigenvalue of the matrix $\mb R$ ($\lambda_{\mb R}>0$ since $\mb R$ is positive definite).
\end{proof}
\subsection{Proof of Lemma \ref{Lemma:L2}}
\label{Subsection:AppB}
\begin{proof}:
	Recall that the input to retraction operator , $\bar{\mathbf{x}}_{(k)}$, is the update of $\mathbf{x}_{(k)}$ on the tangent space, i.e., $\bar{\mathbf{x}}_{(k)}\not\in\mathcal{S}^L$ and hence $|\bar{{x}}_{l(k)}|\geq1 \ \forall \ l$. On the other hand, the output from the retraction step is a constant modulus point, i.e, $ \mathbf{x}_{(k+1)}=\mathbf{R}(\bar{\mathbf{x}}_{(k)})\in \mathcal{S}^L$. Then, $\bar{\mathbf{x}}_{(k)}$ can be written in terms of $\mb x_{(k+1)}$ as 
	\be \label{diffPsi}
	\bar{\mathbf{x}}_{(k)}=\mathbf{x}_{(k+1)}+\Psi \mathbf{x}_{(k+1)}
	\ee
	where $\Psi= \text{diag}(\psi_1,\psi_2, \hdots , \psi_L)$ is non-negative diagonal matrix. Note that Eq (\ref{diffPsi}) is derived from the retraction step, see subsection \ref{equivalence_retraction} of this Appendix. Using Eq (\ref{diffPsi}) the difference between the values of the cost function at $\bar{\mathbf{x}}_{(k)}$ and $\mathbf{x}_{(k+1)}$ will be
	\begin{IEEEeqnarray}{rCl}
		\bar{f}(\bar{\mb x}_{(k)})&-&\bar{f}(\mb x_{(k+1)})\nonumber\\* \quad
		&=&\mathbf{x}_{(k+1)}^H(\mb P +\gamma \mb I)\Psi \mathbf{x}_{(k+1)}\nonumber\\
		&&+\mathbf{x}_{(k+1)}^H\Psi (\mb P +\gamma \mb I) \mathbf{x}_{(k+1)}\nonumber\\
		&&+\mathbf{x}_{(k+1)}^H\Psi (\mb P +\gamma \mb I)\Psi \mathbf{x}_{(k+1)}\nonumber\\
		&&-\mb q^H\Psi \mb x_{(k+1)}-\mb x_{(k+1)}^H\Psi \mb q\nonumber\\
		&=&2 \gamma \mathbf{x}_{(k+1)}^H\Psi \mathbf{x}_{(k+1)}+\mathbf{x}_{(k+1)}^H\big(\Psi \mb P +  \mb P \Psi\big) \mathbf{x}_{(k+1)}\nonumber\\
		&&+\>\mathbf{x}_{(k+1)}^H\Psi (\mb P +\gamma \mb I)\Psi \mathbf{x}_{(k+1)}\nonumber\\
		&&-\mb q^H\Psi \mb x_{(k+1)}-\mb x_{(k+1)}^H\Psi \mb q\nonumber\\
		&\geq& \> 2 \gamma \mathbf{x}_{(k+1)}^H\Psi \mathbf{x}_{(k+1)}+\mathbf{x}_{(k+1)}^H\big(\Psi \mb P+ \mb P\Psi\big) \mathbf{x}_{(k+1)}\nonumber\\
		&&- \>\mb q^H\Psi \mb x_{(k+1)}-\mb x_{(k+1)}^H\Psi \mb q\label{Eq:diff3_1}\\
		&=& \> 2 \gamma\|\Psi \mathbf{x}_{(k+1)}\|_1+\mathbf{x}_{(k+1)}^H\big(\Psi \mb P+ \mb P \Psi\big) \mathbf{x}_{(k+1)}\nonumber\\
		&&- \>\mb q^H\Psi \mb x_{(k+1)}-\mb x_{(k+1)}^H\Psi \mb q\label{Eq:diff3}
	\end{IEEEeqnarray}
	The inequality Eq (\ref{Eq:diff3_1}) holds since $(\mb P +\gamma \mb I)$ is positive semidefinite and the equality Eq (\ref{Eq:diff3}) holds since $\mb x_{(k+1)}^H\Psi \mb x_{(k+1)}=\|\Psi \mb x_{(k+1)}\|_1$ {and this is true since $\mb x_{(k+1)}$ is constant modulus vector}. To go further in this simplification, {and since $\mb P\geq 0$ and $\Psi\geq 0$}, the following theorem about $\big(\Psi \mb P + \mb P \Psi\big)$ (Theorem 7.5 in \cite{zhang2011matrix}) can be utilized
	\begin{IEEEeqnarray}{rCl}
		\IEEEeqnarraymulticol{3}{l}{
		-\frac{1}{4} \lambda_{\Psi} \lambda_{\mb P }\mb I_L\leq\big(\Psi \mb P + \mb P \Psi\big)
	}\label{ineq1_1}\\* \quad
		&\Rightarrow&  -\frac{L}{4} \lambda_{\Psi} \lambda_{\mb P }\leq \mb x_{(k+1)}^H\big(\Psi \mb P + \mb P \Psi\big)\mb x_{(k+1)}\label{ineq1}
	\end{IEEEeqnarray}
where $\lambda_{\Psi}$ and $ \lambda_{\mb P }$ are the largest eigenvalue of $\Psi$ and $\mb P$, respectively. Using Eq (\ref{ineq1}) in Eq (\ref{Eq:diff3}) yields
	\begin{IEEEeqnarray}{rCl}
		\bar{f}(\bar{\mb x}_{(k)})&-&\bar{f}(\mb x_{(k+1)})\nonumber\\* \quad
		&\geq& 2 \gamma\|\Psi \mb x_{(k+1)}\|_1-\frac{L}{4} \lambda_{\Psi} \lambda_{\mb P}\nonumber\\
		&&- \> \mb q^H\Psi \mb x_{(k+1)}-\mb x_{(k+1)}^H\Psi \mb q\label{Eq:diff4_1}\\
		&=& 2 \gamma\|\Psi \mb x_{(k+1)}\|_1-\frac{L}{4} \lambda_{\mb P}\|\Psi \mb x_{(k+1)}\|_{\infty}\nonumber\\
		&&- \>\mb q^H\Psi \mb x_{(k+1)}-\mb x_{(k+1)}^H\Psi \mb q\label{Eq:diff4_2}\\
		& \geq & \> 2 \gamma\|\Psi \mb x_{(k+1)}\|_1-\frac{L}{4} \lambda_{\mb P}\|\Psi \mb x_{(k+1)}\|_1\nonumber\\
		&&- \>\mb q^H\Psi \mb x_{(k+1)}-\mb x_{(k+1)}^H\Psi \mb q\label{Eq:diff4_3}\\
		& \geq &  \>2 \gamma\|\Psi \mb x_{(k+1)}\|_1-\frac{L}{4} \lambda_{\mb P}\|\Psi \mb x_{(k+1)}\|_1\nonumber\\ 
		&&-\> 2\|\Psi \mb x_{(k+1)}\|_2 \|\mb q\|_2\label{Eq:diff4_4}\\
		& \geq &  \>2 \gamma\|\Psi \mb x_{(k+1)}\|_1-\frac{L}{4} \lambda_{\mb P}\|\Psi \mb x_{(k+1)}\|_1\nonumber\\ 
		&&-\> 2\|\Psi \mb x_{(k+1)}\|_1 \|\mb q\|_2\label{Eq:diff4_5}\\
		&=& \> -\Big(\frac{L}{4} \lambda_{\mb P}+2 \|\mb q\|_2\Big)\|\Psi \mb x_{(k+1)}\|_1\nonumber\\
		&& \>+2 \gamma\|\Psi \mb x_{(k+1)}\|_1\nonumber\\
		&=& \> \big(2 \gamma-\frac{L}{4} \lambda_{\mb P}-2 \|\mb q\|_2\big)\|\Psi \mb x_{(k+1)}\|_1\label{Eq:diff4_6}\\
		& \geq &  \> 0 \label{Eq:diff4}
	\end{IEEEeqnarray}
where Eq (\ref{Eq:diff4_1}) holds from Eq (\ref{ineq1}), Eq (\ref{Eq:diff4_2}) holds since $\|\Psi \mb x_{(k+1)}\|_{\infty} = \text{max}_l |\psi_l x_{l(k+1)}|= \lambda_{\Psi}$, Eq (\ref{Eq:diff4_3}) holds since $\|\Psi\mb x_{(k+1)}\|_{\infty}\leq \|\Psi \mb x_{(k+1)}\|_1$, Eq (\ref{Eq:diff4_4}) holds since $\mb q^H\Psi \mb x_{(k+1)}\leq \|\Psi \mb x_{(k+1)}\|_2 \|\mb q\|_2$, Eq (\ref{Eq:diff4_5}) holds since $\|\Psi \mb x_{(k+1)}\|_2 \leq \|\Psi \mb x_{(k+1)}\|_1$, and finally Eq (\ref{Eq:diff4}) holds if 
	\be \label{cond2}
	\gamma \geq \frac{L}{8}\lambda_{\mb P}+\|\mb q\|_2
	\ee
\end{proof}
\subsection{Proof of Lemma \ref{Lemma:L3}}
\label{Subsection:AppC}
\begin{proof}:
	From Lemmas \ref{Lemma:L1} and \ref{Lemma:L2}, we have
	\be \bar{f}(\mb x_{(k)})\geq \bar{f}(\mathbf{x}_{(k+1)}) \Rightarrow \bar{f}(\mb x_{(k)})- \bar{f}(\mathbf{x}_{(k+1)}) \geq0\ee
	Since $\mb x_{(k)}^H\mb x_{(k)}=\mb x_{(k+1)}^H \mb x_{(k+1)}=L$, then 
	\be  \bar{f}(\mb x_{(k)})- \bar{f}(\mathbf{x}_{(k+1)})=f(\mb x_{(k)})- f(\mathbf{x}_{(k+1)})\geq 0
	\ee
	Then the sequence $\{f(\mathbf{x}_{(k)})\}_{k=0}^{\infty}$ is non-increasing and since $f(\mathbf{x})\geq 0, \ \forall \ \mb x$ and hence it is bounded below, then it converges to a finite value $f^{*}$.
\end{proof}
\subsection{The equivalence between Eq (\ref{diffPsi}) and the retraction step}
\label{equivalence_retraction}
In this subsection, we show the equivalence between Eq (\ref{diffPsi}) and the retraction step (Step 5 in Algorithm \ref{Alg:PDR_basic}). This Eq is used to write $\bar{\mathbf{x}}_{(k)}$ in terms of $\mathbf{x}_{(k+1)}$. Starting from the retraction formula, this equivalence can be shown as follows:  
	\begin{IEEEeqnarray}{rCl}
	\mathbf{x}_{(k+1)}&=&\mathbf{R}\big(\bar{\mathbf{x}}_{(k)}\big)=\bar{\mathbf{x}}_{(k)} \odot \frac{1}{|\bar{\mathbf{x}}_{(k)}|}=\bar{\Psi}\bar{\mathbf{x}}_{(k)}\label{Eq:Proof_retraction1}
	\end{IEEEeqnarray}
where $\bar{\Psi}= \text{diag}\Big(\frac{1}{|\bar{x}_{1(k)}|},\frac{1}{|\bar{x}_{2(k)}|}, \hdots , \frac{1}{|\bar{x}_{L(k)}|}\Big)$. 
Solving for $\bar{\mathbf{x}}_{(k)}$ yields 
\begin{IEEEeqnarray}{rCl}
	\bar{\mathbf{x}}_{(k)}&=&\bar{\Psi}^{-1}\mathbf{x}_{(k+1)}\label{Eq:Proof_retraction2}
\end{IEEEeqnarray}

{Since $\bar{\mathbf{x}}_{(k)}\in \mathcal{T}_{\mb x_{(k)}}\mathcal{S}^L$ and hence $|\bar{{x}}_{l(k)}|\geq1, \ l =1,2,\hdots,L$, the quantity $|\bar{{x}}_{l(k)}|$ can be written as $|\bar{{x}}_{l(k)}|=1+\psi_l$ with $\psi_l \geq 0\ \forall \ l$. Using this, the matrix $\bar{\Psi}$ will be 
$\bar{\Psi} =\text{diag}\Big(\frac{1}{1+\psi_1},\frac{1}{1+\psi_2}, \hdots , \frac{1}{1+\psi_L}\Big)\label{matrixPsi1}$, the matrix $\bar{\Psi}^{-1}$ will be 
\begin{IEEEeqnarray}{rCl} 
	\bar{\Psi}^{-1}&=&\text{diag}\Big(1+\psi_1,1+\psi_2, \hdots , 1+\psi_L\Big)\nonumber\\
	&=&\mb I_{L}+\Psi \label{matrixPsi2}
\end{IEEEeqnarray}
where $\Psi= \text{diag}\big(\psi_1,\psi_2, \hdots , \psi_L\big)$. Substituting this value of $\bar{\Psi}^{-1}$ in Eq (\ref{Eq:Proof_retraction2}), the vector $\bar{\mathbf{x}}_{(k)}$ will be 
\begin{IEEEeqnarray}{rCl} 
	\bar{\mathbf{x}}_{(k)}&=&\big(\mb I_{L}+\Psi\big)\mathbf{x}_{(k+1)}\nonumber\\
	&=&\mathbf{x}_{(k+1)}+\Psi\mathbf{x}_{(k+1)}\label{Eq:Proof_retraction3}
\end{IEEEeqnarray}
Eq (\ref{Eq:Proof_retraction3}) hence reduces to Eq (\ref{diffPsi}).}
\scriptsize
\bibliographystyle{IEEEbib}
\bibliography{refs3,refs2}
\end{document}